\documentclass[aps,prb,floatfix,reprint,superscriptaddress]{revtex4-2}
\usepackage{amsmath,amssymb}
\usepackage{amsthm}
\usepackage{graphicx}
\usepackage{epsfig}
\usepackage{psfrag}
\usepackage[usenames]{color}
\usepackage{xcolor}

\usepackage{soul,color}
\usepackage{tabu}
\usepackage{multirow}
\usepackage{bbold}
\usepackage{mathtools}
\usepackage{siunitx}
\usepackage{verbatim}
\usepackage{bbm}
\usepackage{enumitem}
\usepackage{tcolorbox}

\usepackage{array}
\newcolumntype{C}[1]{>{\centering\arraybackslash}p{#1}}

\definecolor{DarkGreen}{RGB}{0,100,0}

\usepackage{hyperref}

\newtheorem{theorem}{Theorem}
\newtheorem{lemma}[theorem]{Lemma}
\newtheorem{definition}{Definition}

\newtheorem{remark}{Remark}

\begin{document}
\title{Microscopic geometric theory for gapped excitations in fractional topological fluids}

\author{Yuzhu Wang}
\email{yuzhu.wang@ntu.edu.sg}
\affiliation{School of Physical and Mathematical Sciences, Nanyang Technological University, 639798 Singapore}


\author{Bo Yang}
\email{yang.bo@ntu.edu.sg}
\affiliation{School of Physical and Mathematical Sciences, Nanyang Technological University, 639798 Singapore}
\date{\today}

\begin{abstract}
We propose a geometric description of all gapped excitations in fractional quantum Hall phases that reveals several fundamental understandings with experimental consequences. These include a duality between the Hilbert space of multiple gapped ``graviton-like" spin 2 excitations at $\nu=1/3$ Laughlin phase, and that of non-Abelian quasiholes of an irrational Haffnian conformal field theory. This leads us to construct microscopic wave functions for multiple higher-spin gapped neutral modes in the Laughlin phase. Carrying spin $s \ge 2$, they emerge from higher-order geometric deformations of the topological ground state, and live within the Gaffnian conformal Hilbert space asymptotically. We show that the full many-body Hilbert space of an FQH phase can be generated from superpositions of such geometric deformations, supporting a concrete geometric interpretation of all gapped excitations. We analyse the scattering between these higher spin modes and conjecture they can have long lifetime, and propose methods for their experimental detections.

\end{abstract}

\maketitle

\hypertarget{sec:i}{\textit{Introduction--}} Strongly interacting topological phases in two dimensions, including fractional quantum Hall (FQH) states in Landau levels and fractional quantum anomalous Hall (FQAH) states in Chern bands, host rich emergent phenomena driven by strong correlations \cite{Klitzing1980,Tsui1982,Park2023,Xu2023,Cai2023}. While the universal properties of the topological ground-state and gapless quasihole manifold are well understood in many settings, the structure and dynamics of gapped excitations remain far less explored, even for idealized model Hamiltonians. This is because analytic solutions are absent, and exponentially large Hilbert-space dimensions hamper numerical computations. Clarifying these excitations is not only a fundamental problem but also important for the experimental realization of topological phases in realistic quantum Hall systems or the more complicated lattice Chern bands: the robustness of these phases depends on the dynamics of these gapped excitations.

\begin{figure}[t]
\begin{center}
\includegraphics[width=\linewidth]{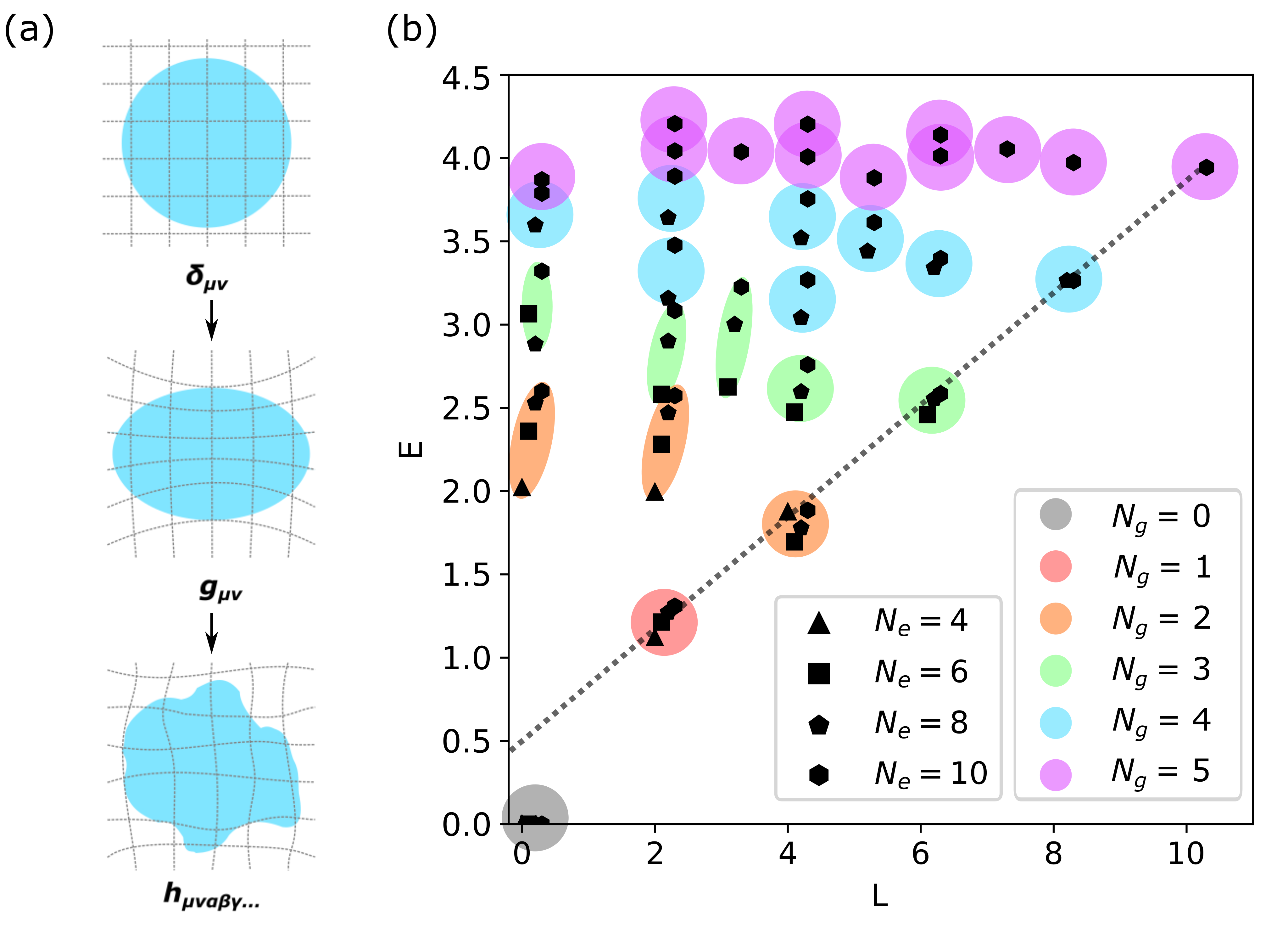}
\caption{\textbf{a. Schematic of ground-state deformations.} Starting from the uniform FQH ground state, the lowest-order area-preserving deformation gives spin-$2$ graviton modes associated with an anisotropic metric $g_{\mu\nu}$. Higher-spin modes arise from higher-order deformations, naturally described by higher-rank tensor structures that cannot be reduced to products of local metrics without redundancy. \textbf{b. Haffnian conformal Hilbert space (CHS) as a multiple Laughlin-graviton space.} $V_1$-resolved spectrum in the Haffnian CHS, defined as the nullspace of a three-body model Hamiltonian at $N_o=3N_e-2$. Markers denote system size, and colors label the number $N_g$ of spin-2 excitations. The highest-$L$ multi-graviton states exhibit linear dispersion (dashed line), consistent with quantized excitations. Exact Haffnian quasihole counting matches that of bosonic spin-2 gravitons, showing that the CHS of non-Abelian anyons can equivalently be spanned by geometric multi-Laughlin-graviton states.}
\label{fig1}
\end{center}
\end{figure}

A major step forward is the recognition that the long-wavelength magnetoroton branch can be viewed as a geometric excitation, the spin-$2$ graviton mode (GM), arising from quantum fluctuations of an emergent guiding-center metric \cite{Yang2012, Golkar2016, Liou2019, wang2021analytic, Yuzhu2023}. This viewpoint is complementary to composite-fermion (CF) and parton descriptions, where neutral excitations are typically interpreted as excitons \cite{PhysRevB.54.4873}. The geometric description also clarifies features not easily captured by the exciton picture, including the numerically found absence of an $L=1$ mode and the qualitative distinction between dipolar and quadrupolar neutral excitations \cite{yang20nematic}. However, the GM represents only a tiny part of the neutral excitations, and its experimentally observed long lifetime, despite being deeply embedded in the continuum of other gapped excitations, remains poorly understood \cite{Liang2024,wang2025dynamicslifetimegeometricexcitations}.

In this Letter, we propose a microscopic geometric theory that provides a unified description of gapped excitations in FQH fluids. We uncover a surprising duality between the conformal Hilbert space (CHS) spanned by non-Abelian Haffnian anyons and the Hilbert space spanned by multiple spin-$2$ gravitons, or the quantum fluctuations of the emergent metric from the Abelian topological ground state. This correspondence directly connects topological anyonic data from an irrational conformal field theory (CFT) to quantum geometric excitations, evoking a condensed-matter analogue of gauge--gravity duality \cite{hooft1993dimensional,susskind1995world}. We further construct microscopic wave functions for a hierarchy of higher-spin (HS) neutral modes for the Laughlin phase using the generators of quantum area-preserving diffeomorphisms, and show they live within the Gaffnian CHS asymptotically. We find that HS modes with small spins may also have long lifetimes due to scarce scattering channels, similar to those of spin-$2$ GMs, and thus propose experimental probes based on photons carrying orbital angular momentum. More generally, we conjecture that all gapped neutral excitations of the Laughlin phase come from a generic geometric deformation of the ground state, and they mediate the effective interaction between charged (gapless) quasihole and (gapped) quasielectrons of the topological phase.

\begin{figure}
\begin{center}
\includegraphics[width=\linewidth]{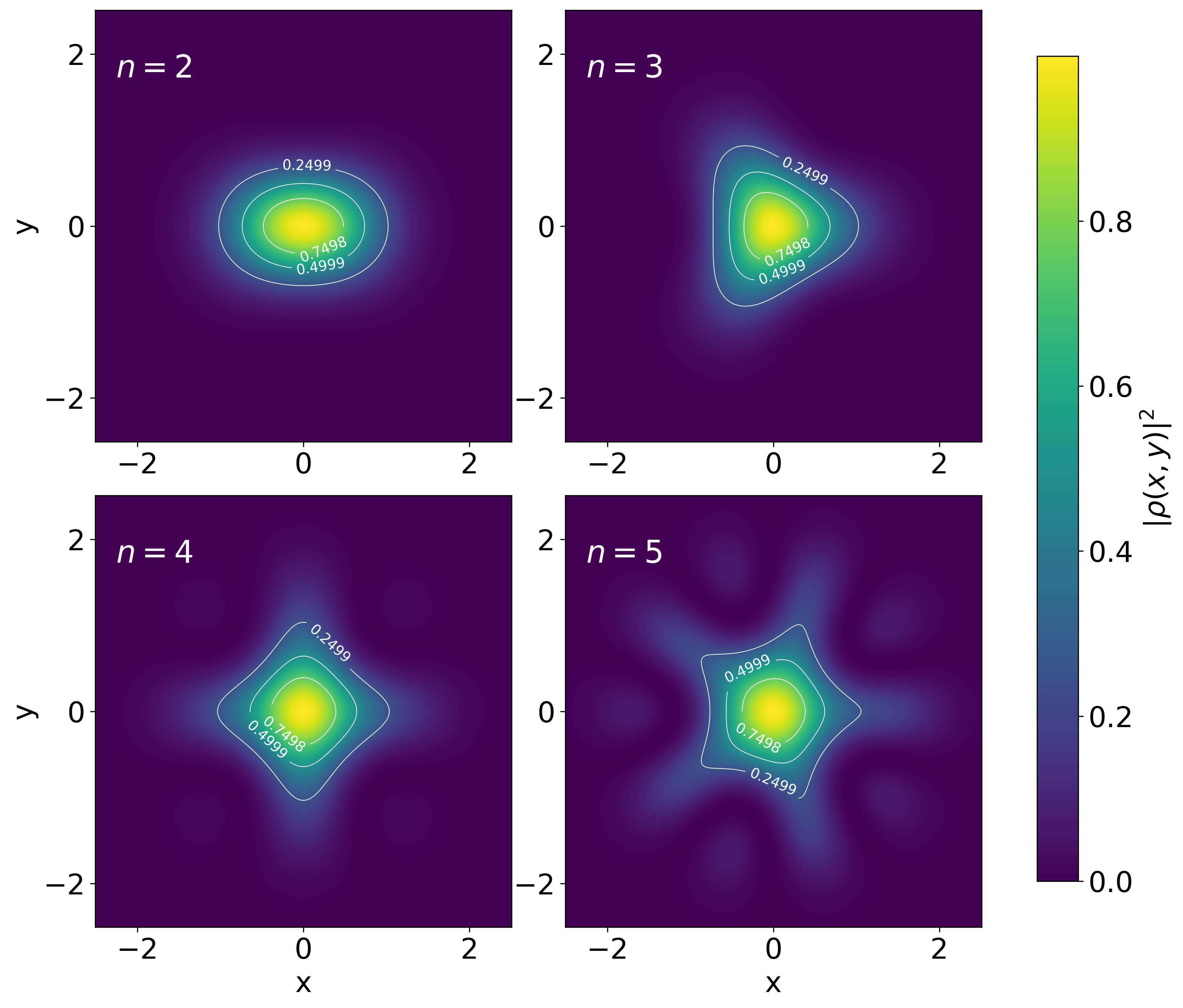}
\caption{\textbf{Density distribution of $\hat{U}_{s}(\alpha) |0\rangle$}. Here, we set the deformation strength $\alpha = 0.2$. Note that there is a $C_s$ symmetry in the wave packet acted by $\hat{U}_{s}$.}
\label{fig2}
\end{center}
\end{figure}

\hypertarget{sec:ii}{\textit{Duality between Laughlin gravitons and Haffnian CHS}--}
To uncover the nature of gapped neutral excitations in a Landau level (LL), it is useful to focus on the smallest nontrivial subspace in which they appear. For the Laughlin state, this role is played by a CHS defined as the nullspace of the Haffnian model Hamiltonian
$\hat{H}_{\mathrm{haf}}=\hat{V}^{(3bdy)}_{3}+\hat{V}^{(3bdy)}_{5}+\hat{V}^{(3bdy)}_{6}$ \cite{PhysRevB.61.10267, PhysRevB.75.075318}. Within this null space at the Laughlin filling $N_o=3N_e-2$, there are many zero-energy states of $\hat{H}_{\mathrm{haf}}$. Two of them are immediately identifiable: the uniform Laughlin ground state and the spin-$2$ Laughlin GM \cite{wang2021analytic}. With short-range interaction, all other states are gapped like the Laughlin GM, and it is natural to ask how to understand those states in the context of the Laughlin phase.

This CHS is known to host four Haffnian quasiholes, whose edge theory corresponds to an irrational $Z_2|Z_2$ orbifold CFT with infinitely many primary fields \cite{PhysRevB.81.115124, PhysRevB.79.045308}. Although its internal structure therefore appears a priori highly nontrivial, Green showed that the CHS can be decomposed into a set of spin-$2$ ``composite bosons'' \cite{green2001strongly}. A more physical interpretation becomes apparent when we resolve the CHS by short-range two-body interactions (e.g., the $V_1$ pseudopotential), as shown in Fig.~\ref{fig1}: the spectrum separates into distinct energy branches, and states near the spectral boundary exhibit an approximately linear dispersion. Each branch can be equivalently labeled by a fixed number $N_g$ of spin-$2$ bosons. Meanwhile, within each total angular momentum sector $L$, the counting of Haffnian quasihole states matches exactly the counting obtained by coupling identical spin-$2$ bosons and imposing bosonic symmetrization, as shown in Fig.~\ref{fig1}. This counting is fixed purely by rotational symmetry and statistics, independent of microscopic dynamics\footnote{Without loss of generality, the spectrum in Fig.~\ref{fig1} can be represented by a Hamiltonian proportional to the boson number $N_b$.}.

We can thus further clarify the physical nature of the previously identified ``composite bosons'' and the energetic structure underlying this correspondence. Surprisingly, we find that the $s>2$ states at given electron numbers have unit overlap with explicitly constructed states consisting of multiple Laughlin gravitons obtained from their root configurations, establishing them as genuine multi-graviton excitations \footnote{Further details about the Haffnian state and the explicit way of constructing multi-Laughlin-graviton modes are given in the Supplementary Materials.}. This identification further implies that the nearly linear dispersion observed in the spectrum arises from quantized \textit{geometric} excitations. Taken together, the mysterious ``composite bosons'' in the Haffnian CHS are precisely Laughlin gravitons.

These results reveal a nontrivial reorganization of the Haffnian CHS: although it is naturally defined in terms of non-Abelian quasiholes associated with an irrational CFT, it admits an exact and complete description in terms of quantized geometric excitations of the abelian Laughlin phase with no finite-size effect. Interestingly, the exact matching of counting implies that gravitons are global collective modes in the FQH phase \cite{yang2025quantum}, with no orbital degrees of freedom (i.e., only zero orbital momentum or angular momentum is allowed). Given that the energy of the graviton depends on the microscopic details and can be either finite or zero \cite{wang2021analytic}, the Hilbert space structure thus suggests gravitons are \emph{massless} quasiparticles in the non-relativistic limit in the context of field theory description. The energy of the graviton comes from the product of a ``divergent speed of light" and a vanishing momentum. Such Hilbert space algebra can also be established for the non-Abelian Moore-Read phase, which will be discussed elsewhere.

\begin{figure}[ht]
\begin{center}
\includegraphics[width=\linewidth]{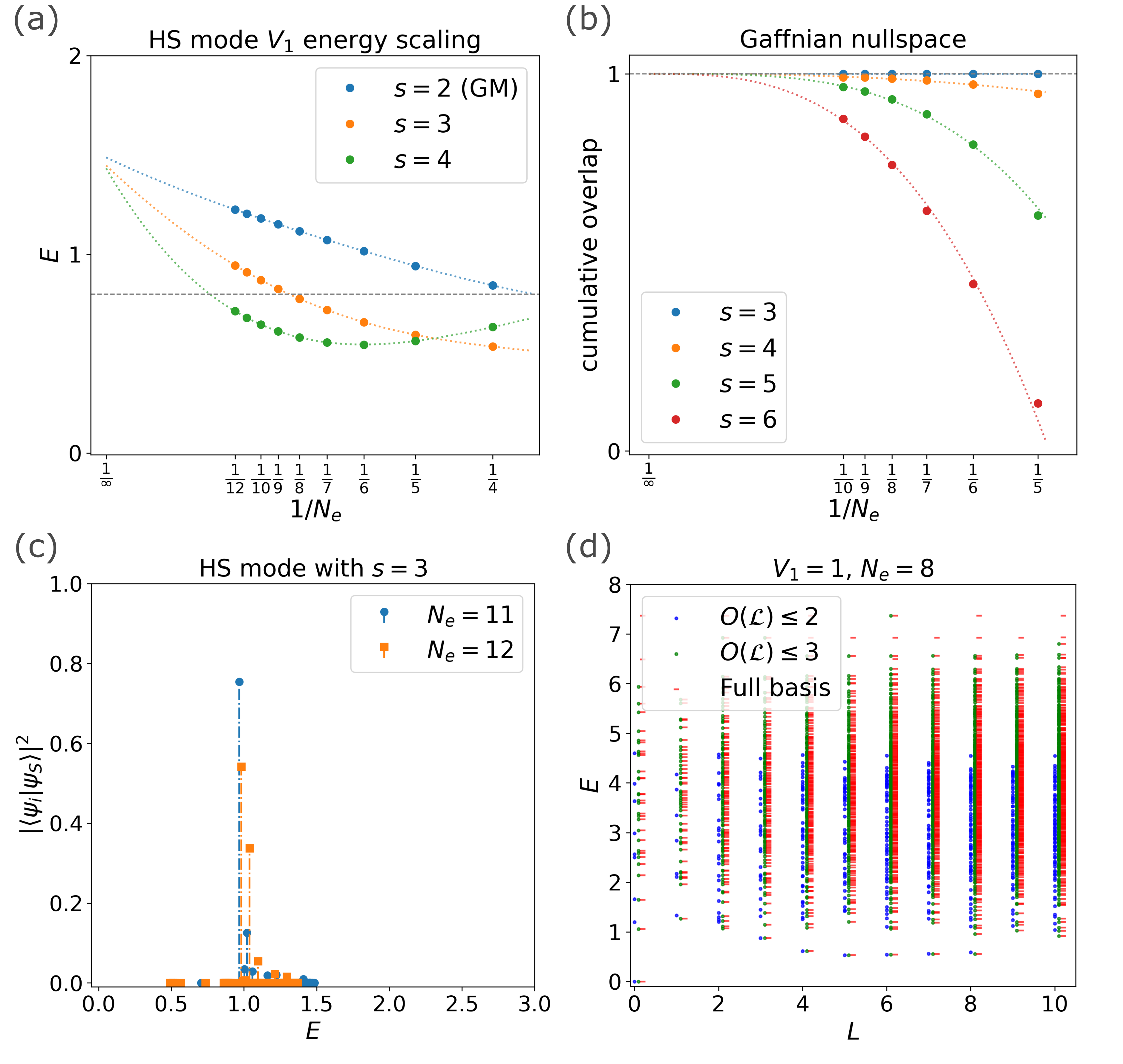}
\caption{\textbf{(a) Finite-size scaling of HS mode energies.} Their energies enter the continuum (above dashed line) and approach the GM energy. While the absolute values are non-universal and depend on system size, the thermodynamic extrapolation indicates convergence to the vicinity of the GM energy.
\textbf{(b) Finite-size scaling of HS modes' cumulative overlaps within Gaffnian nullspace.} All the Laughlin HS modes asymptotically live within this space.
\textbf{(c) Spectral function of HS mode with $s=3$.} Even after entering the continuum, the spectral peak is still sharp, indicating a long lifetime similar to the $s=2$ GM. 
\textbf{(d) Spectrum from deformation basis.} For $N_e=8$, including deformations up to third order, reproduces the full Hilbert space and thus recovers the complete spectrum of $\hat{V}_1$.
}
\label{fig3}
\end{center}
\end{figure}

\hypertarget{sec:iii}{\textit{Higher-spin modes}--}  The exact dualities above suggest that geometric excitations can provide a natural interpretation for gapped excitations beyond just a single ``graviton-like" quadrupole excitation. Since the spin-2 graviton comes from the leading-order ground-state geometric deformation (i.e., from a circle to an ellipse), we now construct higher-spin (HS) neutral modes as higher-order geometric deformations of the ground state (i.e., from a circle to an arbitrary shape), in analogy to the spin-2 graviton modes.

With a magnetic field, it is convenient to separate the coordinates into cyclotron $\tilde{R}$ and guiding-center $\bar{R}$ components, defined by 
$\hat{\tilde{R}}^{a}=\hat{r}^{a}+\epsilon^{ab}\hat{\pi}_{b}$ and $\hat{\bar{R}}^{a}=-\epsilon^{ab}\hat{\pi}_{b}$, where $a,b\in\{x,y\}$, $\hat{r}$ is the real-space position, and $\hat{\pi}$ the covariant momentum:
\begin{equation}
[\hat{r}^{a}_{j}, \hat{\pi}_{jb}] = i \delta^{a}_{b} \delta_{ij}, \quad [\hat{\pi}_{ia}, \hat{\pi}_{jb}] = i\cdot eB \cdot  \delta_{ij} \epsilon_{ab}.
\end{equation}
Here $i$ and $j$ denote the particle indices. For simplicity, we set the magnetic length  $\ell_B = \sqrt{\hbar/eB} = 1$. 
Introducing complex zweibeins $\bar{\omega}_a$ satisfying $\bar{\omega}_a \bar{\omega}^a = 0$, $\bar{\omega}_a^*\bar{\omega}^a=1$ (and similarly for cyclotron degrees of freedom $\tilde{\omega}_a$), we define the guiding-center ladder operators with unimodular Euclidean metrics and compatible complex structures:
\begin{equation}
\hat{b}=\bar{\omega}_{a}\hat{\bar{R}}^{a}, \quad
\hat{b}^{\dagger}=\bar{\omega}_{a}^{*}\hat{\bar{R}}^{a}, \quad [\hat{b}, \hat{b}^{\dagger}] = 1.
\end{equation}

The spin-$2$ graviton is generated by acting with $\sum_{i=1}^{N_e} (\hat{b}_i)^2$ on the topological ground state \cite{Yang2012} \footnote{One can get the graviton with the opposite chirality as $\left|\psi_{2}^{-}\right\rangle\sim \sum_{i=1}^{N_e} (\hat{b}^\dagger_i)^2 \left|\psi_0\right\rangle$. The same idea applies to the HS modes as well.}. This naturally motivates considering the family of operators $(\hat{b})^s$ with $s\ge 2$ as generators of neutral modes carrying higher spins. Acting on the ground state, they define a hierarchy of higher-spin (HS) modes,
\begin{equation}
|\psi_s\rangle \propto \sum_{i=1}^{N_e} (\hat b_i)^s \,|\psi_0\rangle , \quad s\ge 2.
\end{equation}
Their \textit{geometric} meaning is made explicit already at the single-particle level via the unitary transformation
\begin{equation}
\hat U_s(\alpha)=e^{\alpha (\hat{b}^\dagger)^{s} - \alpha^* \hat{b}^{s}}, \quad \alpha\in\mathbb C.
\end{equation}
Without loss of generality, we apply $\hat{U}_{s}(\alpha)$ with a small $\alpha \in \mathbb{R}$ to the 2D coherent state $\langle \boldsymbol{r} | 0 \rangle \propto e^{-r^2}$, with $\boldsymbol{r}=r e^{i\theta_r}$, giving the asymptotic density distribution
\begin{equation}
\begin{aligned}
\rho_{\alpha}^{(s)}(\boldsymbol{r}) & \equiv \langle \boldsymbol{r} |\hat{U}_{s}(\alpha) | 0 \rangle\\
& \sim e^{-|r|^2}\left[1+2 r^s \cos (s \cdot \theta_r) \cdot \alpha \right]+\mathcal{O}\left(\alpha^2\right).
\end{aligned}
\end{equation}
Note that normal ordering with Taylor expansion cannot be directly applied here due to the vanishing radius of convergence at $\alpha=0$ \cite{Braunstein87generalized}. Examples for several $s$ are shown in Fig.~\ref{fig2} where the wave packet $\hat{U}_{s}(\alpha)|0\rangle$ exhibits $C_s$ symmetry. For $s=2$, this deformation corresponds to a quadrupolar metric fluctuation, while for $s>2$ it represents an intrinsically higher-rank geometric distortion that cannot be reduced to metric degrees of freedom. This identifies $\hat b^s$ as generators of higher-spin geometric deformations of the ground state.

For finite systems on the sphere, the HS modes are exactly equivalent to the single-mode approximation (SMA), each of them also carrying an angular momentum $s$ \cite{Yang2012,yang2025quantum}. It is however important to distinguish between HS modes, which are in the long wavelength limit carrying no dipole moment, and other SMA states carrying a finite momentum (and thus a finite dipole moment). This subtle point can be most easily illustrated using spherical geometry. For an SMA state with fixed $s$, its linear momentum scales as $q\sim s/R \propto 1/\sqrt{N_e}\to 0$, where $R$ is the radius of the sphere. Thus, an SMA state is an HS mode only in the thermodynamic limit, when $s/N_e\to 0$. All other SMA states with finite $s/N_e$ in the thermodynamic limit are neutral collective excitations at \textit{finite} momentum \cite{girvin1985collective,Girvin1986}. Thus, in finite-size spectra, only the SMA modes with small $s$ compared to the system size approximate the HS modes well. For these states, we can also see their energy merge into the magnetoroton branch, as shown in Fig.~\ref{fig3}.

In the Abelian Laughlin phase at $\nu=1/3$, the finite-size scaling of $|\psi_{s}\rangle$ with $s\in[2,4]$ under the $V_1$ pseudopotential is shown in Fig.~\ref{fig3}a, where the energies $E_s = \langle \psi_s|\hat{V}_1 |\psi_s\rangle$ exhibit only weak dependence on $s$ as $N_e\to\infty$. Even more remarkably, all the Laughlin HS modes asymptotically reside within the Gaffnian CHS $\mathcal{H}_{\mathrm{gaf}}$, which is the nullspace of the Gaffnian model Hamiltonian
$\hat{H}_{\mathrm{gaf}}=\hat{V}^{(3bdy)}_{3}+\hat{V}^{(3bdy)}_{5}$, since their cumulative overlap with this null space increases with system size and approaches unity, as shown in Fig.~\ref{fig3}b. This behavior is highly nontrivial: as $N_e$ increases, $\mathcal{H}_{\mathrm{gaf}}$ constitutes a progressively smaller fraction of (it is actually of measure zero within) the full Hilbert space, so an increase of overlap onto it is \emph{not} expected. Here, by contrast, the geometric HS modes become \textit{more} sharply confined within $\mathcal{H}_{\mathrm{gaf}}$ as $N_e$ increases. This reflects an emergent thermodynamic property of the long-wavelength neutral sector, where $|\psi_s\rangle$ becomes increasingly well-approximated by the Gaffnian quasihole states.

\hypertarget{sec:v}{\textit{Lifetime and experimental relevance of HS modes} --}
Having established that HS modes are well-defined geometric excitations in the long-wavelength limit, we now examine their dynamical stability and experimental relevance. The HS modes constructed microscopically here provide their explicit realization of neutral quasiparticles carrying spin $s$ in effective field theory language \cite{nguyen2025spinfractionalquantumhall}.

A natural decay process is the scattering of a single HS mode into multiple HS modes. The number of such channels grows rapidly with spin, scaling as $n(s)\sim \frac{1}{4s\sqrt{3}}\,e^{\pi\sqrt{2s/3}}$. However, these processes are strongly suppressed for two reasons. First, there is a universal energy mismatch: the energies of multi–HS-mode states are generically higher than those of single HS modes. Second, the coupling matrix elements between single HS modes and multi–HS-mode states rapidly decrease with system size \footnote{Further details are given in the Supplementary Materials.}. For example, at $s=4$, $|\psi_{4}\rangle$ can only scatter to two spin-2 modes, and we find that their overlap vanishes in the thermodynamic limit while the energy of the two–spin-2 state remains significantly higher than that of $|\psi_{4}\rangle$ \footnote{Similar behavior is expected for $|\psi_{5}\rangle$, and also in other FQH phases, indicating that the stability mechanism does not rely on the statistics of the underlying quasiparticles.}.

Scattering into other excitations outside the HS sector may also occur, and we do not yet have a complete understanding of their coupling strengths. Nevertheless, numerical calculations indicate that such processes remain weak in practice. As shown in Fig.~\ref{fig3}c, even after $|\psi_{3}\rangle$ enters the continuum, its spectral function retains a pronounced and narrow peak, indicating a long lifetime and suggesting that HS modes should be experimentally accessible in much the same way as the GM.




These results motivate the question of how to access HS modes experimentally in the same long-wavelength, angular-momentum–resolved regime where they are sharply defined. Photons are ideal probes of GMs, since they carry spin $s=1$ and negligible linear momentum \cite{Liang2024}. However, accessing HS modes requires higher spin transfer. In principle, this can be achieved by multi-photon processes, where the net spin of several photons is transferred to the system. Yet such processes are inherently weak, as higher-order optical couplings can be strongly suppressed, and are difficult to control experimentally due to decoherence and phase-matching constraints in nonlinear optics \cite{Mollow68two, boyd2008nonlinear}.

A more promising route is to exploit the orbital angular momentum (OAM) of photons, which can be realized with Laguerre--Gaussian (LG) beams \cite{allen92orbital}. An LG mode is labeled by two quantum numbers $(p,l)$, where $p\geq 0$ is the radial index and $l\in\mathbb{Z}$ is the OAM (azimuthal index) \cite{plick15physical}. The transverse profile contains a helical phase factor $e^{i l \phi}$, so each photon carries quantized OAM $\hbar l$ in addition to its spin. This provides a direct and tunable mechanism for exciting HS modes, extending beyond the $s=2$ graviton: because the angular momentum is transferred through the azimuthal phase rather than linear momentum, the in-plane momentum imparted is still negligible, ensuring that the probed states remain in the long-wavelength sector. Moreover, our numerics show that the excitation energies of HS modes are of the same order as those of GMs, well within the range accessible by optical photons.

\hypertarget{sec:vi}{\textit{The Hilbert space of geometric excitations} --} We have focused on representative gapped excitations, such as the GM and HS modes. However, these modes do not exhaust all the gapped excitations of an FQH phase, nor all the possible geometric deformations of the ground state. One may further ask whether \textit{all} gapped excitations in a FQH phase admit a geometric description. This can indeed be achieved through a geometric construction based on the algebra of quantum deformations.

The guiding-center and cyclotron density operators, $\hat{\tilde{\rho}}_{\mathbf{q}}=\sum_{i=1}^{N_e} e^{i q_{a} \hat{\tilde{R}}_{i}^{a}}$ and $\hat{\bar{\rho}}_{\mathbf{q}}=\sum_{i=1}^{N_e} e^{i q_{a} \hat{\bar{R}}_{i}^{a}}$, obey the GMP algebra:
$\left[\hat{\tilde{\rho}}_{\mathbf{q}_{1}}, \hat{\tilde{\rho}}_{\mathbf{q}_{2}}\right]=  - 2 i \sin \left(\frac{1}{2} \epsilon^{ab} q_{1a} q_{2b}\right) \hat{\tilde{\rho}}_{\mathbf{q}_{1}+\mathbf{q}_{2}}$ and 
$\left[\hat{\bar{\rho}}_{\mathbf{q}_{1}}, \hat{\bar{\rho}}_{\mathbf{q}_{2}}\right]=  2 i \sin \left(\frac{1}{2} \epsilon^{ab} q_{1a} q_{2b}\right) \hat{\bar{\rho}}_{\mathbf{q}_{1}+\mathbf{q}_{2}}$. This is the unique closed algebra of projected density operators in $2$D Bloch bands \cite{steve25closed}. According to the SMA, the GMP modes are given by acting with $\hat{\bar{\rho}}_{\mathbf q}$ on a ground state $|\psi_0\rangle$ \cite{girvin1985collective, Girvin1986}.

Meanwhile, within a single LL, the kinetic Hamiltonian takes the form $\hat{H}_k \sim \hat{a}^{\dagger}\hat{a}+1/2$, so $[\hat b,\hat H]=[\hat b^\dagger,\hat H]=0$. This allows the construction of \textit{infinite} operators that act entirely within the guiding-center sector \cite{Cappelli1993, flohr1994infinite}:
\begin{equation}
\hat{\mathcal{L}}_{m,n}\equiv \sum_{i=1}^{N_e}(\hat b_i^\dagger)^m \hat b_i^n,\qquad m,n\ge 0,
\end{equation}
where $N_e$ is the number of electrons. These quasi-local operators obey the quantum algebra of area-preserving diffeomorphisms, known as the $W_\infty$ algebra, which is isomorphic to the GMP algebra \cite{Cappelli1993, flohr1994infinite}:
\begin{equation}
\begin{aligned}
[\hat{\mathcal{L}}_{m,n},\hat{\mathcal{L}}_{k,l}]
= &\sum_{s=1}^{\min(n,k)}\frac{n!k!}{(n-s)!(k-s)!s!}\,
\hat{\mathcal{L}}_{m+k-s,n+l-s} \\
&-(m\leftrightarrow k,\; n\leftrightarrow l),
\qquad m,n,k,l\in\mathbb{N}.
\end{aligned}
\label{W_infty}
\end{equation}
Correspondingly, $\hat{\bar{\rho}}_{\mathbf q}$ serves as a generating function for the operators $\hat{\mathcal{L}}_{m,n}$ through its expansion in $\mathbf q$. The SMA is therefore equivalent to forming superpositions of guiding-center deformations \cite{Haldane2009, Haldane2011}. This provides a formal justification for the geometric interpretation of GM and HS modes introduced earlier.

Deformations in QH systems arise by successive action of $\hat{\mathcal{L}}_{m,n}$ operators on the ground state $|\psi_0\rangle$:
\begin{equation}
|\mathcal{W}^{(k)}_{\{m_i\},\{n_i\}} \rangle \equiv \tilde{\mathcal{L}}^{(k)}_{\{(m_i,n_i)\}} |\psi_0 \rangle = \prod_{i=1}^{N} \hat{\mathcal{L}}_{m_i,n_i} |\psi_0 \rangle,
\label{Wn_states}
\end{equation}
where $k \le N$ denotes the deformation order, set by the number of operators distinct from a single $\hat{b}^\dagger$ or $\hat{b}$. The HS modes $| \psi_s \rangle$ can be identified as $|\mathcal{W}^{(1)}_{\{s\},\{0\}} \rangle$, which can be proved to be highest-weight states on the sphere. In general, $\tilde{\mathcal{L}}^{(k)}_{\{(m_i,n_i)\}}$ neither produces states with good angular momentum quantum numbers nor preserves the $W_\infty$ algebra, but instead belongs to its universal enveloping algebra $\mathcal{U}(W_\infty)$. 

We refer to the modes with $k>1$ in Eq.~\ref{Wn_states} as $k$-th order \textit{general deformations}, and we can define a Hilbert space $\mathcal{H}$ spanned by all the $|\mathcal{W}^{(k)}_{\{m_i\},\{n_i\}} \rangle$ states. One can prove that $\mathcal{H}$ is exactly the Hilbert space consisting of all the gapped excitations. This is because the linear map between the operator sets $\{l_i\}$ and $\{c_j\}$, with $l_i=\hat{\mathcal{L}}_{n+i,n}$ for $n\in[0,2S]$ and $c_j=\hat{c}_{\sigma +j}^\dagger \hat{c}_{\sigma}$ for $\sigma \in[-S,S]$, is represented by a matrix $M$ with $M_{ij}=0$ for $i<j$ and $\det M\neq 0$, hence invertible. So the deformation basis $\{|\mathcal{W}^{(k)}_{\{(m_i,n_i)\}}\rangle\}$ and the Fock basis ($L_z$ eigenbasis) are \textit{isomorphic}. For finite $N_e$, diagonalizing $\hat{H}$ in $\{|\mathcal{W}^{(k)}_{\{(m_i,n_i)\}}\rangle: k \le k_{\text{max}}\}$ resolves excitations by deformation order and yields the full spectrum once $N_{\text{max}}$ exceeds a threshold. We can use the Laughlin phase at the filling $1/3$ as an example to illustrate this argument: Fig.~\ref{fig3}d shows that one can generate the full $V_1$ spectrum with sufficiently large $k$ in the deformation basis.  As $N_e\to\infty$, any subspace spanned by deformations with a finite order $k$ becomes measure zero in the full Hilbert space $\mathcal{H}$, since $\forall k \in\mathbb{Z}_{+}$, $\nu = N_e/N_o \in (0,1]$, $\lim_{N_e\rightarrow\infty} N_o^{2k}/\binom{N_o}{N_e}= 0$.


\begin{acknowledgments}
We would like to thank Duncan Haldane, Steve Simon, Yoshiki Fukusumi, Dam Thanh Son, Kun Yang, Dung Xuan Nguyen, Giandomenico Palumbo, Patricio Salgado-Rebolledo, Eric Bergshoeff, and Ha Quang Trung for the fruitful discussions. This work is supported by the NTU grant for the National Research Foundation, Singapore, under the NRF fellowship award (NRF-NRFF12-2020-005), and Singapore Ministry of Education (MOE) Academic Research Fund Tier 3 Grant (No. MOE-MOET32023-0003) “Quantum Geometric Advantage.”
\end{acknowledgments}

\bibliography{HS_references}{}

\newpage
\appendix
\counterwithin{figure}{section}

\setcounter{table}{0}
\onecolumngrid

\begin{center}
    {\large \textbf{Supplementary Materials for ``Microscopic geometric theory of gapped excitations in fractional quantum Hall fluids''}}  
\end{center}

This Supplementary Material provides additional details and technical derivations supporting the results presented in the main text. 
Section~\ref{sm:haffnian} reviews the microscopic and conformal-field-theory descriptions of the Haffnian state, including its filling factor, ground-state counting, wave function, and pairing interpretation. 
Section~\ref{sm:multigraviton_construct} describes the construction of neutral excitations from root configurations and squeezing rules, and explains how higher-spin and multi-graviton modes arise microscopically in model quantum Hall states. 
Section~\ref{sm:counting_symNs} analyzes the counting of highest-weight states of $N$ identical spin-$s$ bosons, which provides the algebraic basis for understanding the counting of multi-graviton Hilbert spaces. 
Section~\ref{sm:higher_spin_deform} discusses single-particle deformations that generate higher-spin structures in the lowest Landau level. 
Section~\ref{sm:hs_lifetime} presents additional numerical results for the lifetime and scattering of higher-spin modes.
Section~\ref{sm:higher_spin_algebra} develops the algebraic structure of the higher-spin generators and their relation to the guiding-center $W_\infty$ algebra on the sphere. 
Finally, Section~\ref{sm:all_gapped_are_geometric} proves the geometric interpretation of the neutral excitations discussed in the main text.

\section{The Haffnian state: microscopic and conformal field theoretical construction\label{sm:haffnian}}

The Haffnian state plays a central role in the construction of the conformal Hilbert space discussed in the main text. In this section, we briefly review its microscopic and conformal field theoretical descriptions. We summarize the basic properties of the Haffnian state, including its filling factor, model wave function, ground-state degeneracy, and model Hamiltonian, and discuss the associated $d$-wave pairing picture. We then outline its representation within conformal field theory, which shows an irrational nature and helps with understanding the nontrivial counting of ground states and quasihole states. 

\subsection{Microscopic approach}

The bosonic Haffnian is a paired quantum Hall state at filling factor $\nu=\frac{1}{2}$, and is most naturally interpreted as a $d$-wave paired state of spinless composite bosons \cite{Green2001,PhysRevB.61.10267}. On a genus-$0$ manifold, the densest ground state occurs at orbital number \cite{Green2001,Hermanns2011}
\begin{equation}
N_o = 2N_e - 3,
\end{equation}

A standard form of the Haffnian wave function is
\begin{equation}
\Psi_{\mathrm{Hf}}(\{z_i\})
=
\mathrm{Hf} \left(\frac{1}{(z_i-z_j)^2}\right)
\prod_{i<j}(z_i-z_j)^2,
\label{eq:Haffnian-wavefunction}
\end{equation}
where the Haffnian of a symmetric matrix $M_{ij}$ is defined by
\begin{equation}
\mathrm{Hf}(M_{ij})
=
\sum_{\text{pairings}}
M_{i_1 j_1} M_{i_2 j_2} \cdots M_{i_{N_e/2} j_{N_e/2}}.
\end{equation}

To see the $d$-wave pairing interpretation of the Haffnian \cite{Green2001,PhysRevB.61.10267}, one should recall that in momentum space, the $d$-wave pairing potential transforms as
\begin{equation}
\Delta(\mathbf{k}) \propto (k_x - i k_y)^2,
\end{equation}
so the pair angular momentum is $\ell=-2$ \cite{PhysRevB.61.10267}. In real space, this corresponds to the pair wave function $g(z)\sim 1/z^2$, precisely the structure appearing in the Haffnian prefactor of Eq.~\ref{eq:Haffnian-wavefunction}. In this sense, the Haffnian is the bosonic $d$-wave analog of the $p$-wave Moore-Read Pfaffian. However, unlike the Pfaffian, the Haffnian does not represent a stable weak-pairing topological phase. Rather, the analysis of paired states in two dimensions suggests that the corresponding $d$-wave problem is naturally associated with criticality \cite{PhysRevB.61.10267}. In Green's discussion, the Haffnian is therefore better viewed as a critical paired state, plausibly lying at a transition rather than deep inside an incompressible phase \cite{Green2001}. We will return to this point below. Furthermore, a fermionic wave function can be given by multiplying Eq.~\ref{eq:Haffnian-wavefunction} by a Jastrow factor, and the filling factor becomes $\nu=\frac{2}{6}$ (to distinguish it from the Laughlin state at effectively the same filling). Its densest ground state occurs at the orbital number:
\begin{equation}
N_o = 3N_e - 4.
\end{equation}
So it has a different topological shift from the Laughlin state. 

The Haffnian state is the unique densest zero-energy ground state of a three-body model Hamiltonian \cite{Green2001, PhysRevB.75.075318}, which may be written as:
\begin{equation}
\hat{H}_{\mathrm{Hf}}
=
c_0 \hat{V}^{(3bdy)}_0
+
c_2 \hat{V}^{(3bdy)}_2
+
c_3 \hat{V}^{(3bdy)}_3, 
\label{eq:Haffnian-parent_bosonic}
\end{equation}
where the coefficients $c_0$, $c_2$ and $c_3$ are real numbers, and the nullspace of this model Hamiltonian defines the Haffnian conformal Hilbert space. $\hat{V}^{(3bdy)}_L$ denotes the three-body pseudopotentials equivalent to $\hat{P} \left(\frac{3N_\phi}{2}-L\right)$, with $N_\phi$  the flux number and $\hat{P}(L)$ projecting three particles onto total angular momentum $L$ \cite{Green2001}. For the fermionic Haffnian state, the model Hamiltonian becomes:
\begin{equation}
\hat{H}_{\mathrm{Hf}}
=
c_3 \hat{V}^{(3bdy)}_3
+
c_5 \hat{V}^{(3bdy)}_5
+
c_6 \hat{V}^{(3bdy)}_6.
\label{eq:Haffnian-parent_fermionic}
\end{equation}
The Haffnian wave function has no support in these forbidden triplet sectors, which is why it is annihilated by $H_{\mathrm{Hf}}$ and emerges as the densest zero mode  \cite{Green2001}. In this sense, it seems similar to other model fractional quantum Hall (FQH) states with exact model Hamiltonians. 

However, its physical nature from other well-known states (such as the Laughlin states) is rather different, and this will be exposed when we look at the ground state degeneracy $D_{\text{haf}}$ of the Haffnian state on a torus: for even numbers of bosons or fermions, $D_{\text{haf}} = N_e + 8$ or $N_e/2 + 4$ (without the three-fold center-of-mass degeneracy), which is \textit{increasing} with the system size. This implies that the Haffnian state is critical or gapless in the thermodynamic limit, and also suggests that the Haffnian quasihole counting will be nontrivial and increasing with the system size as well. In fact, such counting can be calculated either from exact diagonalizations with the model Hamiltonian, or from squeezing root configurations with restrictions if we choose specific geometries as introduced below.

The explicit polynomial form of the Haffnian wave function on \textit{a disk or a sphere} is consistent with the generalized exclusion rule based on root configurations \cite{Green2001,Hermanns2011}, although it cannot be represented by a Jack polynomial. The root configuration for the bosonic Haffnian ground state is
\begin{equation}
2\,0\,0\,0\,2\,0\,0\,0\,2\,\ldots\,2\,0\,0\,0\,2,
\end{equation}
which obeys the generalized exclusion rule that there are at most two particles in any four consecutive orbitals, with the \textit{additional rule} that patterns like $0\,2\,0\,0\,1\,0$ are also allowed for providing the correct quasihole counting \cite{Hermanns2011}. For the fermionic case, the root configuration reads:
\begin{equation}
1\,1\,0\,0\,0\,0\,1\,1\,0\,0\,0\,0\,\ldots\,1\,1\,0\,0\,0\,0\,1\,1.
\end{equation}
The generalized exclusion rules give a compact way to generate the Haffnian zero-mode space and quasihole counting: the quasihole sector of the Haffnian can be generated by squeezing (i.e., bringing a pair of electrons towards each other and thus decreasing their relative angular momentum by $2$, as we will introduce in more detail in the next section) from the corresponding root configurations, with the multiplicity in each angular-momentum sector determined by the highest-weight condition. This construction reproduces the zero-mode counting of the model Hamiltonian $\hat{H}_{\mathrm{Hf}}$ and leads to the quasihole counting formula on a disk or a sphere:
\begin{equation}
\sum_b
\binom{b-2+N_{\mathrm{qh}}/2}{b}
\binom{(N_e-b)/2+N_{\mathrm{qh}}}{N_{\mathrm{qh}}}.
\end{equation}
Note that the summation over $b$ is bounded only by $N_e$ rather than solely by the quasihole number $N_{\mathrm{qh}}$. In particular, such quasihole degeneracy \textit{cannot} be interpreted in terms of a finite set of anyon types \cite{Hermanns2011,PhysRevB.84.115121}. This is one of the clearest signatures that the Haffnian is an \textit{irrational} state rather than a stable rational topological phase.

\subsection{Conformal field theory description}

The Haffnian state also admits a representation in terms of chiral conformal field theory correlators. In this formulation, the neutral sector is generated by a chiral bosonic current
\begin{equation}
J(z)=\partial\varphi(z),
\qquad h_J=1,
\end{equation}
whose multi-point correlators reproduce the Haffnian pairing structure,
\begin{equation}
\langle J(z_1)\cdots J(z_{N_e})\rangle
=
\mathrm{Hf} \left(\frac{1}{(z_i-z_j)^2}\right).
\end{equation}
Combining this neutral correlator with the standard vertex operators of the charge sector yields the Haffnian ground-state wave function. In this sense, the analytic structure of the Haffnian can be interpreted as arising from a conformal theory consisting of a $U(1)$ charge mode together with a neutral bosonic current sector.

Quasihole excitations are described by twist operators that modify the boundary conditions of the neutral current sector. Denoting such an operator by $\sigma$, its operator product expansion with the current takes the form \cite{PhysRevB.79.045308}
\begin{equation}
J(z)\,\sigma(0)\sim z^{-1/2}\,\tau(0)+\cdots ,
\end{equation}
which implies that the twist operator carries conformal weight $h_\sigma=\frac{1}{8} $. This structure identifies the neutral sector with a $\mathbb{Z}_2$ orbifold of a $c=1$ boson. As in other orbifold conformal field theories, twist operators introduce branch cuts for the bosonic field and therefore generate multiple conformal blocks when several quasiholes are present, leading to non-Abelian quasiparticle degeneracies.

More detailed information about the quasiparticle content can be obtained from the vertex-algebra construction of the $Z_2 \times Z_2$ state. Within this framework, quasiparticles are classified by the pattern-of-zeros data together with their electric charge and scaling dimension. One finds three distinct families of quasiparticles \cite{PhysRevB.84.115121}: the first family corresponds to the trivial sector and represents the vacuum modulo electrons. The second family contains nontrivial twist-field excitations with scaling dimensions analogous to those appearing in an Ising-type theory, although their electric charge differs from that of the Moore–Read quasihole. The third family is more exotic: its quasiparticle operators depend on a \textit{continuous} parameter that enters the operator-product coefficients of the vertex algebra. Because this parameter is not fixed by the pattern of zeros, the theory contains an \textit{infinite} set of distinct quasiparticle operators belonging to this family.

The fusion structure reflects this enlarged quasiparticle spectrum. While the fusion rules for the first two families resemble those of familiar non-Abelian theories, the presence of the continuously parameterized third family implies that the fusion algebra does not close on a finite set of quasiparticle types. Instead, the vertex algebra generates infinitely many quasiparticle sectors. This infinite fusion structure is the conformal-field-theory manifestation of the \textit{irrational} character of the Haffnian state: unlike rational conformal field theories describing conventional non-Abelian quantum Hall phases, the quasiparticle content of the Haffnian does not truncate to a finite set of primary fields. Consequently, the associated quasihole degeneracies and torus ground-state degeneracies grow with system size rather than remaining finite \cite{PhysRevB.79.045308, PhysRevB.81.115124, PhysRevB.84.115121}.

It is important to emphasize that this pathology is not due to nonunitarity. In contrast to the Gaffnian state described by a nonunitary conformal theory, the Haffnian neutral sector is irrational but \textit{unitary}. Thus, the failure of the Haffnian to describe a stable topological phase arises from the presence of an infinite-dimensional fusion structure rather than from states with zero or negative norm in the representation space \cite{PhysRevB.103.115102}.

\section{Construction of multi-graviton states based on root configurations\label{sm:multigraviton_construct}}

In this section, we describe a systematic method for constructing neutral excitation wave functions using the root-configuration and squeezing framework. This approach provides an efficient way to generate model wave functions for collective modes such as the magnetoroton modes, including the multi-Laughlin-graviton states discussed in the main text.

On a disk or a sphere, many FQH model wave functions are naturally organized using root configurations together with the squeezing principle \cite{bernevig2008model}. A \textit{root configuration} specifies a representative occupation pattern of Landau-level orbitals, from which the full many-body Hilbert space sector is generated by \textit{squeezing} operations that preserve particle number and total angular momentum. This construction encodes the generalized exclusion rule (or generalized Pauli principle) characterizing the underlying topological phase. For example, the Laughlin state at filling $\nu=1/3$ obeys the rule that no more than one electron can occupy any three consecutive orbitals, leading to the root configuration
\begin{equation}
1\,0\,0\,1\,0\,0\,1\,0\,0\,1\,0\,0\,1\cdots .
\end{equation}
All basis states contributing to the model wave function can be obtained from this root configuration through squeezing operations that move particles closer together while preserving their center of mass.

The same framework can be used to describe neutral excitations. In particular, the magnetoroton mode of the Laughlin and the Moore–Read states can be generated by modifying the root configuration locally near one pole of the sphere while leaving the bulk structure unchanged. These root configurations define a restricted Hilbert space whose ground state under the corresponding model Hamiltonian reproduces the neutral collective excitation.

Building on this idea, one can construct a family of root configurations that generate magnetoroton modes as well as other neutral modes. For the Laughlin state at filling $\nu=1/3$, representative configurations are
\begin{equation}
\begin{aligned}
1\,1&\,0\,0\,0\,0\,1\,0\,0\,1\,0\,0\,1\,0\,0\,1\,0\,0\,1 \cdots, \quad L=2 \quad \text{[Magnetoroton]} \\
1\,1&\,0\,0\,0\,1\,0\,0\,0\,1\,0\,0\,1\,0\,0\,1\,0\,0\,1 \cdots, \quad L=3 \quad \text{[Magnetoroton]} \\
1\,1&\,0\,0\,0\,1\,0\,0\,1\,0\,0\,0\,1\,0\,0\,1\,0\,0\,1 \cdots, \quad L=4 \quad \text{[Magnetoroton]} \\
1\,1&\,0\,0\,1\,0\,0\,0\,0\,1\,0\,0\,1\,0\,0\,1\,0\,0\,1 \cdots, \quad L=4 \\
1\,1&\,0\,0\,0\,1\,0\,0\,1\,0\,0\,1\,0\,0\,0\,1\,0\,0\,1 \cdots, \quad L=5 \quad \text{[Magnetoroton]} \\
1\,1&\,0\,0\,1\,0\,0\,0\,1\,0\,0\,0\,1\,0\,0\,1\,0\,0\,1 \cdots, \quad L=5 \\
1\,1&\,0\,0\,0\,1\,0\,0\,1\,0\,0\,1\,0\,0\,1\,0\,0\,1\,0 \cdots, \quad L=6 \quad \text{[Magnetoroton]} \\
1\,1&\,0\,0\,1\,0\,0\,0\,1\,0\,0\,1\,0\,0\,1\,0\,0\,0\,1 \cdots, \quad L=6 \\
1\,1&\,0\,0\,1\,0\,0\,1\,0\,0\,0\,1\,0\,0\,0\,1\,0\,0\,1 \cdots, \quad L=6 \\
& \quad\qquad\qquad \vdots
\end{aligned}
\end{equation}
A characteristic feature of these configurations is the presence of two consecutive electrons localized near the north pole, while the remaining orbitals continue to obey the generalized Pauli principle of the Laughlin ground state. This structure can be interpreted as a pair of quasi-electrons forming a localized neutral excitation on top of the incompressible background.

To remove redundancy associated with the spherical geometry, we impose the highest-weight condition $L=L_z$, selecting only the highest-weight representative of each angular-momentum multiplet. This substantially reduces the dimensionality of the Hilbert space without eliminating any physically distinct states.

Our construction proceeds by temporarily removing the two electrons localized at the north pole and solving the model Hamiltonian in the reduced Hilbert space. Specifically, we annihilate the two orbitals with the largest angular momenta with the operator $\hat c_Q \hat c_{Q-1}$, where $Q$ denotes the highest orbital index (equivalently, the largest $L_z$ quantum number on the sphere). We then determine the ground state of the Haldane pseudopotential Hamiltonian $\hat V_1$ within this reduced Hilbert space. For the magnetoroton modes' root configurations, a uniquely defined ground state always exists, and restoring the two electrons yields the model wave function for the higher-spin excitation. The resulting state satisfies
\begin{equation}
\hat V_1 \hat c_Q \hat c_{Q-1} |\psi_S\rangle =0,
\qquad
\hat L_+ |\psi_S\rangle =0,
\end{equation}
which ensures that it is both a zero mode of the model Hamiltonian and a highest-weight state of total angular momentum. For other modes, one can project out the magnetoroton modes and other excitations with the same to get a uniquely defined wave function. The dispersion of these modes is shown in Fig.~\ref{supp_fig1} More details on this procedure can be found in Ref\cite{Yang2012, yang14nature}. One can also use such procedures for non-Abelian states, which we will leave for future work to address.

\begin{figure}[t]
\begin{center}
\includegraphics[width=0.6\linewidth]{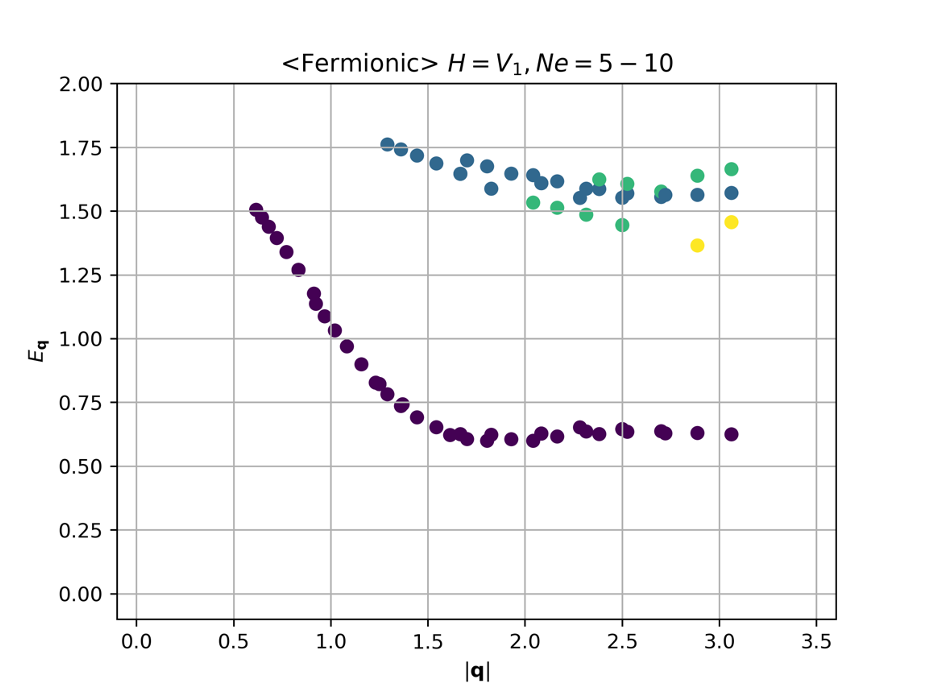}
\caption{\textbf{Dispersion of neutral gapped excitations with respect to the $V_1$ pseudopotential.} 
The magnetoroton branch is clearly resolved, and its overlap with the exact diagonalization eigenstates is unity. The other neutral excitations, shown in lighter colors, are labeled according to their squeezing order: modes with lighter colors have root configurations that can be obtained by squeezing those shown in darker colors. These states lie at higher energies and are therefore deeply embedded in the continuum.}
\label{supp_fig1}
\end{center}
\end{figure}

The same procedure can be extended to construct multi-Laughlin-graviton modes. By observing the representative root configuration for a Laughlin graviton on the north pole is $1\,1\,0\,0\,0\,0$. We can construct the root configurations that contain multiple such localized neutral structures near the pole while the rest of the system retains the Laughlin-type root pattern. The root configurations of multi-Laughlin-graviton states can thus be written as:
\begin{equation}
\begin{aligned}
1\,1&\,0\,0\,0\,0\,1\,1\,0\,0\,0\,0\,1\,0\,0\,1\,0\,0\,1\,0\,0\,1\,0\,0\,1 \cdots, \quad L=4 \quad \text{[2-graviton]}\\
1\,1&\,0\,0\,0\,0\,1\,1\,0\,0\,0\,0\,1\,1\,0\,0\,0\,0\,1\,0\,0\,1\,0\,0\,1 \cdots, \quad L=6 \quad \text{[3-graviton]},\\
1\,1&\,0\,0\,0\,0\,1\,1\,0\,0\,0\,0\,1\,1\,0\,0\,0\,0\,1\,1\,0\,0\,0\,0\,1 \cdots, \quad L=8 \quad \text{[4-graviton]},\\
\end{aligned}
\end{equation}
These roots can still lead to a single, uniquely defined state in each case. Applying the reduced-Hilbert-space procedure described above with the condition $\hat V_1 \hat{c}_Q \hat{c}_{Q-1} \hat{c}_{Q-6} \hat{c}_{Q-7}|\psi_S\rangle =0$, therefore, yields a well-defined microscopic wave function for the corresponding multi-graviton excitation. One can also put the root component $1\,1\,0\,0\,0\,0$ symmetrically on both poles to generate $L=0$ states. The important message is that we surprisingly found that \textit{the $L=4$ double-graviton mode has unity overlap with the Haffnian quasihole states from exact diagonalization}. This confirms that the duality between the two Hilbert spaces is not just based on counting, but also on the wave functions.

An important implication of this construction is that quantum Hall gravitons are not free particles, since the energy cost of a double graviton mode is lower than the sum of two single graviton modes: for example, with $N_e=9$ and Haldane pseudopotential $V_1$, a single Laughlin graviton mode has the energy $E_{1G}=1.295$, while the double graviton mode has the energy $E_{_2G} = 2.304$. The admissible root configurations impose strong constraints on the Hilbert space, and the resulting multi-graviton states cannot be interpreted simply as products of independent single-graviton modes. Instead, they form correlated collective excitations whose structure is encoded directly in the root configurations.

Finally, we note that the uniqueness of the resulting wave function requires a minimum system size. Below a certain threshold, different root configurations may generate overlapping Hilbert spaces, and the corresponding zero modes are no longer uniquely defined. The minimal system size required for the double-graviton mode to be uniquely determined is $N_e=9$, and this size minimum increases with the number of gravitons.

\section{Counting highest-weight states of identical spinful bosons\label{sm:counting_symNs}}

In this section, we discuss how to generate the counting for the highest-weight states of identical spinful bosons, which supports our calculation of the dimension of the Hilbert space spanned by a collection of spin-$2$ Laughlin graviton modes, and also works for more general cases. 

We consider $N$ identical bosons carrying an internal $\mathrm{SU}(2)$ spin-$s$ degree of freedom. The one-particle Hilbert space is the irreducible representation $V_s$ of dimension $2s+1$, whose weights (magnetic quantum numbers) are $m=-s,-s+1,\dots,s$. The $N$-boson spin Hilbert space is the totally symmetric tensor power
\begin{equation}
\mathcal{H}_{N,s} \equiv \mathrm{Sym}^N(V_s).
\end{equation}
Our goal is to decompose $\mathcal{H}_{N,s}$ into irreducible $\mathrm{SU}(2)$ representations,
\begin{equation}
\mathrm{Sym}^N(V_s) \cong \bigoplus_{J=0}^{Ns} m_J\, V_J,
\end{equation}
where $V_J$ denotes the spin-$J$ irrep and $m_J$ its multiplicity. The quantity
\begin{equation}
\#\mathrm{HW}(N,s) \equiv \sum_J m_J
\end{equation}
counts the number of highest-weight multiplets appearing in $\mathrm{Sym}^N(V_s)$.

To determine the multiplicities, we work with $\mathrm{SU}(2)$ characters. Writing a group element in terms of the class angle $\theta$ as $\mathrm{diag}(e^{i\theta/2},e^{-i\theta/2})$, the character of the spin-$J$ irrep is
\begin{equation}
\chi_J(\theta)=\frac{\sin \big((2J+1)\theta/2\big)}{\sin(\theta/2)}.
\label{eq:su2_character}
\end{equation}
For the spin-$s$ representation the weights are $m\in\{-s,-s+1,\dots,s\}$, so the character is
\begin{equation}
\chi_s(\theta)=\sum_{m=-s}^{s} e^{i m\theta}.
\end{equation}

The character of the symmetric power $\mathrm{Sym}^N(V_s)$ is conveniently obtained from the generating function
\begin{equation}
\mathcal{Z}_s(t;\theta)
\equiv \sum_{N\ge 0} t^N\,\chi_{\mathrm{Sym}^N(V_s)}(\theta).
\end{equation}
Because bosonic symmetrization allows arbitrary occupation of each weight $m$, the contribution of a fixed weight forms a geometric series $(1-t e^{i m\theta})^{-1}$. Summing over all weights gives
\begin{equation}
\mathcal{Z}_s(t;\theta)
= \prod_{m=-s}^{s}\frac{1}{1-t e^{i m\theta}}.
\label{eq:sym_power_gen}
\end{equation}
Thus $\chi_{\mathrm{Sym}^N(V_s)}(\theta)$ is obtained as the coefficient of $t^N$ in the expansion of \ref{eq:sym_power_gen}. In practice this coefficient can be extracted by expanding the product in powers of $t$.

The multiplicities $m_J$ follow from character orthogonality. For class functions of $\theta\in[0,2\pi)$, the $\mathrm{SU}(2)$ Haar inner product reads
\begin{equation}
\langle f,g\rangle
=
\frac{1}{\pi}\int_{0}^{2\pi} d\theta\,\sin^2(\theta/2)\,f(\theta)\,g(\theta),
\end{equation}
under which the characters are orthonormal, $\langle\chi_J,\chi_{J'}\rangle=\delta_{J,J'}$. Projecting onto $\chi_J$ therefore gives
\begin{equation}
m_J
=
\left\langle
\chi_{\mathrm{Sym}^N(V_s)},\chi_J
\right\rangle
=
\frac{1}{\pi}\int_{0}^{2\pi} d\theta\,\sin^2(\theta/2)\,
\chi_{\mathrm{Sym}^N(V_s)}(\theta)\chi_J(\theta).
\label{eq:multiplicity_projection}
\end{equation}
Numerically, the integral in \ref{eq:multiplicity_projection} can be evaluated by discretizing $\theta$ on a uniform grid and replacing the integral by a Riemann sum; the resulting values are then rounded to the nearest integer.

Finally, once the multiplicities $\{m_J\}$ are obtained, the number of highest-weight multiplets is simply
\begin{equation}
\#\mathrm{HW}(N,s)=\sum_{J=0}^{Ns} m_J.
\end{equation}
This quantity counts the number of independent highest-weight states in $\mathrm{Sym}^N(V_s)$, which corresponds to the number of irreducible $\mathrm{SU}(2)$ multiplets appearing in the symmetric $N$-boson Hilbert space. The countings for spin-$2$, $3$, and $4$ are shown in Table.~\ref{tab:spin2_counting}, \ref{tab:spin3_counting} and \ref{tab:spin4_counting}.

\begin{table}[t]
\centering
\renewcommand{\arraystretch}{1.5}
\setlength{\tabcolsep}{10pt}
\begin{tabular}{c|ccccccccccc}
\hline\hline
$N \backslash L$ & 0 & 1 & 2 & 3 & 4 & 5 & 6 & 7 & 8 & 9 & 10 \\
\hline
1 & 0 & 0 & 1 & 0 & 0 & 0 & 0 & 0 & 0 & 0 & 0 \\
2 & 1 & 0 & 1 & 0 & 1 & 0 & 0 & 0 & 0 & 0 & 0 \\
3 & 1 & 0 & 1 & 1 & 1 & 0 & 1 & 0 & 0 & 0 & 0 \\
4 & 1 & 0 & 2 & 0 & 2 & 1 & 1 & 0 & 1 & 0 & 0 \\
5 & 1 & 0 & 2 & 1 & 2 & 1 & 2 & 1 & 1 & 0 & 1 \\
6 & 2 & 0 & 2 & 1 & 3 & 1 & 3 & 1 & 2 & 1 & 1 \\
\hline\hline
\end{tabular}
\caption{Counting of highest-weight states for $N$ identical spin-2 bosons. The table entries give the multiplicity of the $\mathrm{SU}(2)$ irrep with total angular momentum $L$ in the totally symmetric space $\mathrm{Sym}^N(V_2)$. For the Laughlin phase, the largest number of graviton modes is given by $\lfloor N_e/2 \rfloor$.}
\label{tab:spin2_counting}
\end{table}

\begin{table}[h]
\centering
\renewcommand{\arraystretch}{1.5}
\setlength{\tabcolsep}{8pt}
\begin{tabular}{c|ccccccccccccccccccc}
\hline\hline
$N \backslash L$ 
& 0 & 1 & 2 & 3 & 4 & 5 & 6 & 7 & 8 & 9 & 10 & 11 & 12 & 13 & 14 & 15 & 16 & 17 & 18 \\
\hline
1 
& 0 & 0 & 0 & 1 & 0 & 0 & 0 & 0 & 0 & 0 & 0 & 0 & 0 & 0 & 0 & 0 & 0 & 0 & 0 \\
2 
& 1 & 0 & 1 & 0 & 1 & 0 & 1 & 0 & 0 & 0 & 0 & 0 & 0 & 0 & 0 & 0 & 0 & 0 & 0 \\
3 
& 0 & 1 & 0 & 2 & 1 & 1 & 1 & 1 & 0 & 1 & 0 & 0 & 0 & 0 & 0 & 0 & 0 & 0 & 0 \\
4 
& 2 & 0 & 2 & 1 & 3 & 1 & 3 & 1 & 2 & 1 & 1 & 0 & 1 & 0 & 0 & 0 & 0 & 0 & 0 \\
5 
& 0 & 2 & 1 & 4 & 2 & 4 & 3 & 4 & 2 & 3 & 2 & 2 & 1 & 1 & 0 & 1 & 0 & 0 & 0 \\
6 
& 3 & 0 & 4 & 3 & 6 & 3 & 7 & 4 & 6 & 4 & 5 & 2 & 4 & 2 & 2 & 1 & 1 & 0 & 1 \\
\hline\hline
\end{tabular}
\caption{Counting of highest-weight states for $N$ identical spin-3 bosons. The table entries give the multiplicity of the $\mathrm{SU}(2)$ irrep with total angular momentum $L$ in the totally symmetric space $\mathrm{Sym}^N(V_3)$.}
\label{tab:spin3_counting}
\end{table}

\begin{table}[h]
\centering
\renewcommand{\arraystretch}{1.5}
\setlength{\tabcolsep}{8pt}
\begin{tabular}{c|ccccccccccccccccccc}
\hline\hline
$N \backslash L$ 
& 0 & 1 & 2 & 3 & 4 & 5 & 6 & 7 & 8 & 9 & 10 & 11 & 12 & 13 & 14 & 15 & 16 & 17 & 18 \\
\hline
1 
& 0 & 0 & 0 & 0 & 1 & 0 & 0 & 0 & 0 & 0 & 0 & 0 & 0 & 0 & 0 & 0 & 0 & 0 & 0 \\
2 
& 1 & 0 & 1 & 0 & 1 & 0 & 1 & 0 & 1 & 0 & 0 & 0 & 0 & 0 & 0 & 0 & 0 & 0 & 0 \\
3 
& 1 & 0 & 1 & 1 & 2 & 1 & 2 & 1 & 1 & 1 & 1 & 0 & 1 & 0 & 0 & 0 & 0 & 0 & 0 \\
4 
& 2 & 0 & 3 & 1 & 4 & 2 & 4 & 2 & 4 & 2 & 3 & 1 & 2 & 1 & 1 & 0 & 1 & 0 & 0 \\
5 
& 2 & 1 & 4 & 3 & 6 & 5 & 7 & 5 & 7 & 5 & 6 & 4 & 5 & 3 & 3 & 2 & 2 & 1 & 1 \\
6 
& 4 & 1 & 7 & 5 & 11 & 7 & 13 & 9 & 13 & 10 & 12 & 8 & 11 & 7 & 8 & 5 & 6 & 3 & 4 \\
\hline\hline
\end{tabular}
\caption{Counting of highest-weight states for $N$ identical spin-4 bosons. The table entries give the multiplicity of the $\mathrm{SU}(2)$ irrep with total angular momentum $L$ in the totally symmetric space $\mathrm{Sym}^N(V_4)$, truncated to $L \le 18$.}
\label{tab:spin4_counting}
\end{table}

\section{Higher-spin deformations of coherent states\label{sm:higher_spin_deform}}

In the main text, we introduced the unitary transformation
\begin{equation}
\hat U_s(\alpha)=e^{\alpha (\hat b^\dagger)^s-\alpha^*\hat b^s},
\end{equation}
which generates higher-spin deformations of a coherent state. Here we briefly justify two statements used in the main text: (i) the perturbative expansion of matrix elements around $\alpha=0$ has zero radius of convergence for $s>2$, and  
(ii) the form of the leading small-$\alpha$ deformation of the density.

Consider the vacuum matrix element
\begin{equation}
F_s(\alpha)=\langle 0| \hat U_s(\alpha) |0\rangle .
\end{equation}
Expanding formally in powers of $\alpha$ gives a series of the form
\begin{equation}
F_s(\alpha)
=
\sum_{n=0}^{\infty}
a_n \alpha^{2n},
\qquad
a_n \sim (-1)^n \frac{(ns)!}{(2n)!} C_n ,
\label{Fs_series}
\end{equation}
where $C_n$ grows approximately exponentially with $n$ \cite{Braunstein87generalized}.  
The radius of convergence $R$ is determined by
\begin{equation}
\frac{1}{R}
=
\limsup_{n\to\infty}|a_n|^{1/(2n)} .
\end{equation}

The dominant growth arises from the factorial ratio $(ns)!/(2n)!$.  
Using Stirling's approximation and taking the $2n$-th root gives
\begin{equation}
\left(\frac{(ns)!}{(2n)!}\right)^{1/(2n)}
\sim
\mathrm{const}\times n^{(s-2)/2}.
\end{equation}
For $s>2$, this diverges as $n\to\infty$, implying $R=0$. Thus, the Taylor expansion around $\alpha=0$ has zero radius of convergence.  
This explains why a naive normal-ordered expansion cannot be used to define the operator in perturbation theory, as noted in the main text. And we also confirmed this with numerical calculations of a truncated series and various $\alpha$ parameters in Eq.~\ref{Fs_series}.

Despite the nonconvergent perturbative series, the leading deformation of the wavefunction for small $\alpha$ can be obtained directly from the operator definition.  
Expanding the unitary operator to first order,
\begin{equation}
\hat U_s(\alpha)
=
1+\alpha (\hat b^\dagger)^s-\alpha^*\hat b^s+\mathcal{O}(\alpha^2),
\end{equation}
and acting on the oscillator ground state gives
\begin{equation}
\hat U_s(\alpha)|0\rangle
=
|0\rangle+\alpha(\hat b^\dagger)^s|0\rangle+\mathcal{O}(\alpha^2).
\end{equation}

In the coherent-state representation $\langle \boldsymbol r|0\rangle\propto e^{-r^2}$ with $\boldsymbol r=r e^{i\theta_r}$, the operator $(\hat b^\dagger)^s$ generates the angular harmonic $r^s e^{i s\theta_r}$.  
Consequently
\begin{equation}
\langle \boldsymbol r|\hat U_s(\alpha)|0\rangle
\propto
e^{-r^2}
\Big[
1+\alpha r^s e^{is\theta_r}
+\alpha^* r^s e^{-is\theta_r}
\Big]
+\mathcal{O}(\alpha^2).
\end{equation}
For real $\alpha$ this becomes
\begin{equation}
\langle \boldsymbol r|\hat U_s(\alpha)|0\rangle
\sim
e^{-r^2}
\left[
1+2\alpha r^s\cos(s\theta_r)
\right]
+\mathcal{O}(\alpha^2).
\end{equation}

Therefore the leading effect of $\hat U_s(\alpha)$ is an $s$-fold angular modulation of the Gaussian density profile,
\begin{equation}
\rho^{(s)}_\alpha(\boldsymbol r)
\sim
e^{-r^2}
\left[
1+2\alpha r^s\cos(s\theta_r)
\right],
\end{equation}
which is the expression used in the main text.  This deformation explicitly carries angular momentum $s$, providing the geometric interpretation of $(\hat b)^s$ as generators of higher-spin neutral modes.

Note that one can also use Pad\'e approximants to get the results above, which replace the divergent power series by a rational function whose Taylor expansion reproduces the first several coefficients. This allows $F_s(\alpha)$ to be analytically continued to finite $\alpha$ and yields well-defined probability distributions.

\section{Lifetime and mutual scattering of higher-spin modes\label{sm:hs_lifetime}}

In this section, we present additional numerical results for higher-spin modes. First, according to the method introduced in Section.~\ref{sm:counting_symNs}, we can calculate the dimension of the Hilbert space spanned by a collection of higher-spin modes with different spins, and if we assume that there is no relative angular momenta between them (as we did for the Hilbert space spanned by spin-$2$ graviton modes), the counting is explicitly shown in Table.~\ref{Table1}. After carefully examining the counting, we find that the Laughlin higher-spin modes only span a \textit{proper subspace} of the Gaffnian conformal Hilbert space.

\begin{table}[h]
\centering
\renewcommand{\arraystretch}{1.5}
\setlength{\tabcolsep}{10pt}
\begin{tabular}{c|cccccccc}
    \hline\hline
          $N_e$ & $\mathbf{3}$ & $\mathbf{4}$ & $\mathbf{5}$ & $\mathbf{6}$ & $\mathbf{7}$ & $\mathbf{8}$ & $\mathbf{9}$ & $\mathbf{10}$ \\
         \hline\hline
          $\mathcal{H}_{\text{haf}}$& 2 & 5 & 5 & 10  & 10 & 18 & 18 & 30 \\
         \hline
         $\mathcal{H}_{\text{HS}}$& 3 & 7 & 8 & 18  & 19 & 33 & 42 & 61 \\
         \hline
         $\mathcal{H}_{\text{gaf}}$& 3 & 8 & 17 & 45  & 113 & 313 & 862 & 2465 \\
         \hline\hline
    \end{tabular}
    \caption{\textbf{The counting of states in the Haffnian CHS $\mathcal{H}_{\text{haf}}$, the Gaffnian CHS $\mathcal{H}_{\text{gaf}}$, and the subspace $\mathcal{H}_{\text{HS}}$ spanned by Laughlin HS modes for the filling $N_o = 3N_e-2$.} The counting within $\mathcal{H}_{\text{HS}}$ can be calculated based on Sec.\ref{sm:counting_symNs}.The counting implies that $\mathcal{H}_{\text{haf}} \subset \mathcal{H}_{\text{HS}} \subset \mathcal{H}_{\text{HS}}$.}
    \label{Table1}
\end{table}

In Fig.~\ref{fig:lifetime}, we present the spectral functions of the Laughlin higher-spin modes for several system sizes. The $s=2$ graviton mode exhibits a sharp resonance peak in all systems shown. The $s=3$ and $s=4$ modes also display very sharp peaks for the system sizes considered here, as they remain below the continuum. For the higher-spin modes at larger $s$, we observe an interesting additional feature: besides the main peak, each of them develops a subleading resonance peak at an energy close to that of the graviton mode. This behavior is consistent with our finite-size scaling analysis of the higher-spin excitation energies, which suggests that in the thermodynamic limit, all higher-spin modes become almost energetically degenerate with the graviton mode.

\begin{figure}[h]
\begin{center}
\includegraphics[width=\linewidth]{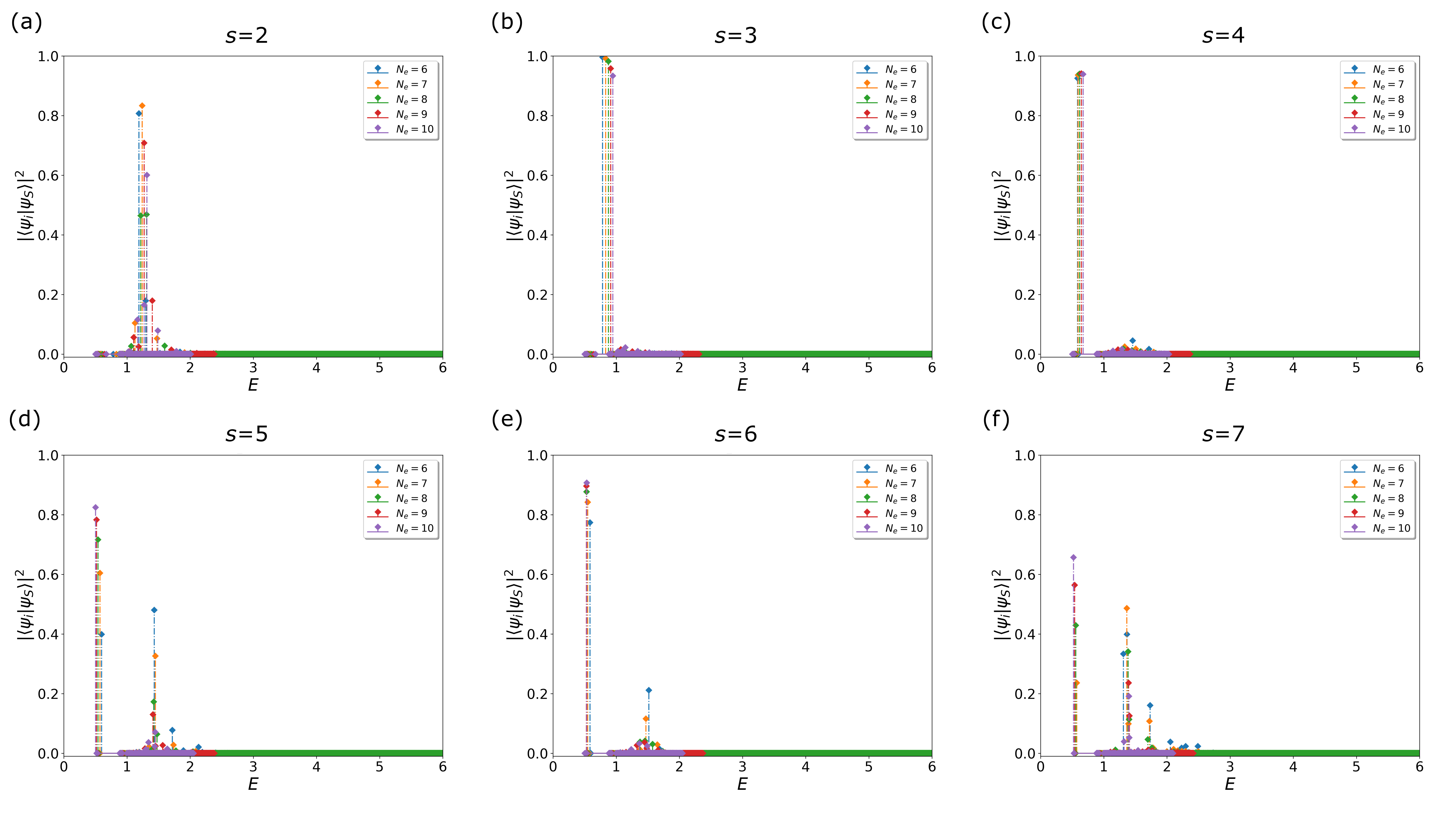}
\caption{\textbf{(a)-(f) Spectral functions of higher spin modes with $s = 2$ to $7$.} }
\label{fig:lifetime}
\end{center}
\end{figure}

In Fig.~\ref{fig:scattering}, we show the scattering amplitudes between, on the left, the higher-spin mode $\psi_4$, the spin-$4$ magnetoroton, and the spin-$2+2$ mode generated by $\sum_{i=1}^{N_e} \hat{b}_i^2 |\psi_2 \rangle$, and, on the right, the higher-spin mode $\psi_5$, the spin-$5$ magnetoroton, and the spin-$2+3$ mode generated by $\sum{i=1}^{N_e} \hat{b}_i^2 |\psi_3 \rangle$. The results show clearly that, in finite systems, the higher-spin modes (or more precisely, the single-mode-approximation states with finite wave vectors) are never identical to the corresponding magnetoroton modes, while the overlap between the spin-$4$/$5$ higher-spin mode and the spin-$2+2$/$2+3$ mode remains small. Interestingly, the magnetoroton mode always lies within the subspace spanned by the other two modes.

\begin{figure}[ht]
\begin{center}
\includegraphics[width=\linewidth]{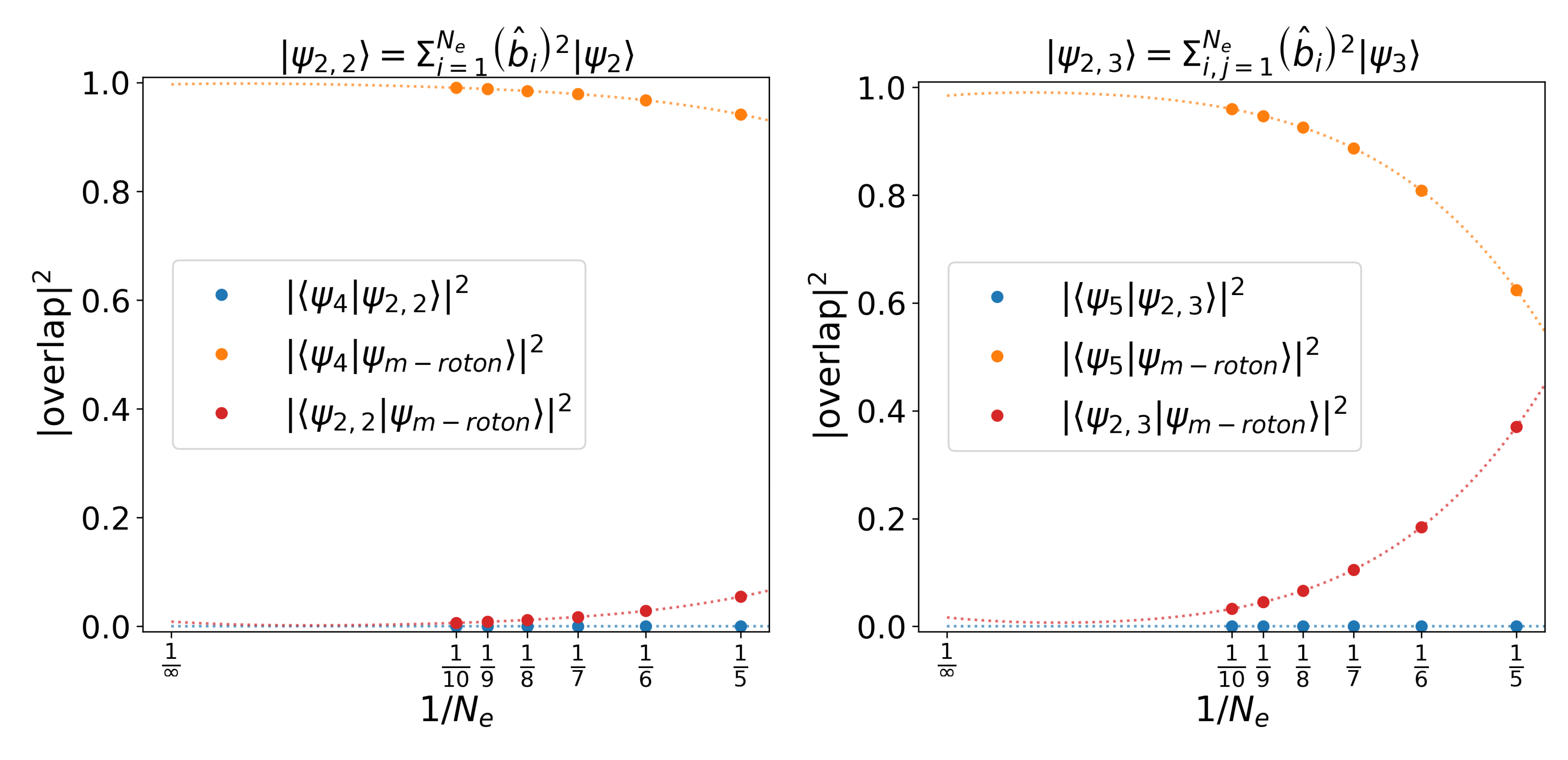}
\caption{\textbf{Scattering amplitude between different states.} Here ``m-roton'' denotes the magnetoroton.. $psi_s$ denotes the higher-spin modes with spin $s$.}
\label{fig:scattering}
\end{center}
\end{figure}

\section{\texorpdfstring{Higher-spin operators and the $W_\infty$ algebra}{Higher-spin operators and the Winfty algebra}\label{sm:higher_spin_algebra}}

In this section, we will show more details about the higher-spin operators, including their representations on the disk and the sphere geometry, and their algebraic properties.

We first derive Eq.~(7) of the main text. Starting from the first-quantized single-particle operators:
$$
\mathcal L_{m,n} = (b^\dagger)^m b^n, \qquad m,n \in \mathbb{N}, \qquad [b, b^\dagger] = 1.
$$

\textbf{Lemma.} For any non-negative integers $p,q$, we have \cite{Pain_comm}
\begin{equation}
\left[ b^p, (b^\dagger)^q \right]
= \sum_{s=1}^{\min(p,q)} \frac{p! q!}{(p-s)! (q-s)! s!} (b^\dagger)^{q-s} b^{p-s}.
\label{eq:commutator_identity}
\end{equation}

We compute:
\begin{equation}
\begin{aligned}
\left[ \mathcal L_{m, n}, \mathcal L_{k,l} \right]
&= \left[ (b^\dagger)^m b^n,  (b^\dagger)^k b^l \right] = (b^\dagger)^m \left[ b^n, (b^\dagger)^k \right] b^l 
 - (b^\dagger)^k \left[ b^l, (b^\dagger)^m \right] b^n.
\end{aligned}
\end{equation}

Using equation \eqref{eq:commutator_identity} with $p = n$, $q = k$:
\begin{equation}
\left[ b^n, (b^\dagger)^k \right]
= \sum_{s=1}^{\min(n,k)} \frac{n! k!}{(n-s)! (k-s)! s!} (b^\dagger)^{k-s} b^{n-s}.
\end{equation}

Multiplying on the left by $(b^\dagger)^m$ and on the right by $b^l$, we get:
\begin{equation}
(b^\dagger)^m (b^\dagger)^{k-s} = (b^\dagger)^{m+k-s}, \qquad 
b^{n-s} b^l = b^{n+l-s},
\end{equation}
so the first term becomes:
\begin{equation}
\sum_{s=1}^{\min(n,k)} 
\frac{n! k!}{(n-s)! (k-s)! s!}   \mathcal{L}_{m+k-s,  n+l-s}.
\end{equation}

Swapping $(m,n) \leftrightarrow (k,l)$ gives the second term:
\begin{equation}
\sum_{s=1}^{\min(m,l)} 
\frac{m! l!}{(m-s)! (l-s)! s!}   \mathcal{L}_{k+m-s,  l+n-s}.
\end{equation}

Combining both terms, we obtain the $W_\infty$ algebra:
\begin{equation}
\left[ \mathcal{L}_{m,n}, \mathcal{L}_{k,l} \right] 
= \sum_{s=1}^{\min(n,k)} \frac{n! k!}{(n-s)! (k-s)! s!}   \mathcal{L}_{m+k-s,  n+l-s} 
- (m \leftrightarrow k,  n \leftrightarrow l ).
\label{eq:disk-algebra}
\end{equation}
This algebra clearly retains the same form for multi-particle systems since the ladder operators for different particles commute.

Now we derive the representation of this algebra on the sphere. For monopole strength $N_\phi = 2S$, we define
\begin{equation}
\hat{\mathcal{L}}^{S}_{m,n}
= \sum_{s=-S}^{S} \mathcal{C}_{m,n}(s)  \hat{c}^{\dagger}_{ s+m-n} \hat{c}_{ s},
\qquad
\mathcal{C}_{m,n}(s)
= C_{S, s-n;  m, m}^{S, s-n+m}  C_{S, s;  n, -n}^{S, s-n},
\label{eq:def-Lmn-sphere}
\end{equation}
with $C_{j_1,m_1; j_2,m_2}^{J,M}$ the Clebsch–Gordan coefficients, and with the understanding that any term with $|s|>S$ or $|s+m-n|>S$ is zero. The equivalence of $\hat{\mathcal{L}}_{m,n}$ operators and $\hat{\mathcal{L}}^{S}_{m,n}$ can be verified numerically: by acting them on the corresponding ground state and removing the geometry-dependent single-particle normalization factor, one can see that the overlap between the remaining states is unity.

By using $[\hat{c}^{\dagger}_{a}\hat{c}_{b}, \hat{c}^{\dagger}_{c}\hat{c}_{d}]
= \delta_{b,c} \hat{c}^{\dagger}_{a}\hat{c}_{d} - \delta_{a,d} \hat{c}^{\dagger}_{c}\hat{c}_{b}$, we obtain
\begin{align}
\bigl[\hat{\mathcal{L}}^{S}_{m,n}, \hat{\mathcal{L}}^{S}_{k,l}\bigr]
&= \sum_{s,s'=-S}^{S}\mathcal{C}_{m,n}(s) \mathcal{C}_{k,l}(s') 
\Bigl(\delta_{s+m-n, s'} \hat{c}^{\dagger}_{ s+m-n+k-l}\hat{c}_{ s}
- \delta_{s, s'+k-l} \hat{c}^{\dagger}_{ s'+m-n}\hat{c}_{ s'}\Bigr) \notag\\
&= \sum_{s=-S}^{S}
\Bigl[\mathcal{C}_{m,n}(s) \mathcal{C}_{k,l}(s+m-n)
- \mathcal{C}_{k,l}(s) \mathcal{C}_{m,n}(s+k-l)\Bigr] 
\hat{c}^{\dagger}_{ s+m-n+k-l}\hat{c}_{ s}.
\label{eq:raw-comm}
\end{align}
Eq.~\ref{eq:raw-comm} shows that the algebra is closed on the sphere: the result is again a one–body operator with shift $(m-n)+(k-l)$, so the algebra closes exactly within the space of operators $\hat{\mathcal L}^{S}_{m,n}$.  
Eq.~\ref{eq:raw-comm} therefore proves that the generators form a closed algebra on the sphere.

Because the single-particle Hilbert space contains only $2S+1$ orbitals, the algebra is finite-dimensional at fixed $S$. The structure constants acquire curvature and finite-size corrections, so the algebra is not identical to the planar guiding–center ladder operator algebra. Instead, the spherical algebra should be viewed as a finite-$N_\phi$ realization of the guiding-center $W_\infty$ algebra, i.e.,
\begin{equation}
\text{sphere algebra}  \xrightarrow[]{S\to\infty}  W_\infty ,
\end{equation}
while at finite $S$ the algebra represents a curvature-deformed, finite-size truncation of the planar guiding-center algebra.

Because of degeneracies in the irreducible representations of $SO(3)$, it suffices to focus on the highest-weight (HW) states $|L,L\rangle$. Since the ground state $|\psi_0 \rangle = |0, 0\rangle$ is an HW state, and the angular momentum ladder operators are special cases of $\mathcal{L}$-operators: $\hat{L}_{+} = \hat{\mathcal{L}}^S_{1,0}$,  $\hat{L}_{-} =\hat{\mathcal{L}}^S_{0,1}$, one immediately sees that \textit{spin-$1$ modes vanish}.  By noticing that it is sufficient and necessary for $\left[ \hat{\mathcal{L}}^S_{1,0}, \hat{\mathcal{L}}^S_{m,n} \right] = 0$ to hold by setting $n=0$, one can also prove that \textit{all the higher-spin modes are HW states}.

\section{Geometric construction of neutral gapped excitations\label{sm:all_gapped_are_geometric}}

In this section, we show how to define the neutral excitations in an FQH droplet and prove that any neutral excitation can be constructed as a product of density operators in any finite-dimensional Hilbert space.

First, we consider the finite-dimensional case with fermions and a disk geometry without loss of generality (Note that particle statistics and geometry do not matter in our discussion, so all conclusions apply to bosons and other geometries automatically). We will assume that certain physical conventions and quantities are a priori known, such as particles and orbitals, creation and annihilation operators, commutation relations, Fock states/ground states, etc., without explicit definitions. Although we intend to make the proof formal to prepare for the infinite-dimensional case, the idea is simply to prove the existence of an invertible linear transformation for any finite-dimensional case.

\begin{definition}
A many-body Hilbert space $\mathcal{H}_{N_e, N_o}$ is the one spanned by all Fock states with the filling $(N_e, N_o)$, i.e., $N_e$ particles with $N_o$ orbitals:
\begin{equation}
| \phi_i \rangle = \prod_{j=1}^{N_e} \hat{c}^\dagger_{k_{i,j}} \ |vac\rangle, \ N_e \in \mathbb{Z}_{+}, N_e \le N_o \in \mathbb{Z}_{+}, 
\end{equation}
where $\hat{c}^\dagger_{k}$ denotes the creation operator at $k$ orbital, obeying the canonical anticommutation relations (CARs) and the orbital index $0 \le k_{i,j}\le N_o-1$.
\end{definition}

$\mathcal{H}_{N_e, N_o}$ is the Hilbert space containing all the many-body states with the same number of electrons and fluxes as the ground state in a fractional quantum Hall (FQH) phase, and we can define the excitations in such a space as:
\begin{definition}
With the ground state of an FQH phase $| \psi_0 \rangle \in \mathcal{H}_{N_e,N_o}$, a neutral gapped excitation is a quantum state $| \psi \rangle \in \mathcal{H}_{N_e,N_o}$ that obeys $\langle \psi_0 | \psi \rangle = 0$.
\end{definition}
Note that it is possible for an FQH phase to host multiple ground states, in which case the neutral excitations will be orthogonal to the subspace spanned by all the ground states.

\begin{theorem}
$\forall | \phi_i \rangle = \prod_{j=1}^{N_e} \hat{c}^\dagger_{k_{i,j}} \ |vac\rangle \in \mathcal{H}_{Ne,No}$, define $\mathcal{C}_{a,b} = \hat{c}^\dagger_{a+b} \hat{c}_a$, where $a, b \in \mathbb{Z}$, $\max \{0,-b\} \le a \le \min \{ N_o, N_o - b\}$,
\begin{equation}
    \left\{ \prod_{j=1}^{N_{max}} \mathcal{C}_{a_j,b_j} | \phi_i \rangle| N_{max} \in \mathbb{Z}, 1 \le N_{max} \le N_e\right\} \cong \mathcal{H}_{N_e, N_o}.
\end{equation}
\end{theorem}

\begin{definition}[Fixed-particle-number Hilbert space]
Fix integers $N_o\ge 1$ and $0\le N_e\le N_o$. Let $\{\hat c_k^\dagger,\hat c_k\}_{k=0}^{N_o-1}$ be fermionic creation/annihilation operators obeying the canonical anticommutation relations. The many-body Hilbert space at filling $(N_e,N_o)$, denoted by $\mathcal{H}_{N_e,N_o}$, is the $N_e$-particle subspace of the Fock space generated by these operators, i.e.,
\begin{equation}
\mathcal{H}_{N_e,N_o}
:= \mathrm{span}\Big\{
|\phi_{K}\rangle
:= \hat c_{k_1}^\dagger \hat c_{k_2}^\dagger \cdots \hat c_{k_{N_e}}^\dagger |vac\rangle
\ \Big|\ 
0\le k_1<k_2<\cdots<k_{N_e}\le N_o-1
\Big\}.
\end{equation}
\end{definition}

\begin{definition}[Neutral excitations]
Let $|\psi_0\rangle\in\mathcal{H}_{N_e,N_o}$ be a normalized ground state of an FQH phase (at fixed $(N_e,N_o)$). The neutral excitation subspace is defined as the orthogonal complement of the ground-state subspace. In the non-degenerate case, a (neutral) gapped excitation is any state
\begin{equation}
|\psi\rangle\in\mathcal{H}_{N_e,N_o}
\quad \text{such that}\quad
\langle \psi_0|\psi\rangle=0.
\end{equation}
If the ground state is $g$-fold degenerate and spans $\mathcal{G}\subset\mathcal{H}_{N_e,N_o}$, then a neutral excitation is any state $|\psi\rangle\in\mathcal{H}_{N_e,N_o}$ satisfying
\begin{equation}
\langle \psi_g|\psi\rangle=0 \quad \text{for all}\quad |\psi_g\rangle\in\mathcal{G},
\end{equation}
equivalently, $|\psi\rangle\in\mathcal{G}^\perp$.
\end{definition}

\begin{lemma}[Action of $\mathcal{C}_{a,b}$ on a Slater determinant]
Let
\begin{equation}
|\phi\rangle
= \hat c^\dagger_{k_1}\hat c^\dagger_{k_2}\cdots \hat c^\dagger_{k_{N_e}}|vac\rangle,
\qquad
0\le k_1<k_2<\cdots<k_{N_e}\le N_o-1,
\end{equation}
and define $\mathcal{C}_{a,b}=\hat c^\dagger_{a+b}\hat c_a$ with $a,a+b\in\{0,1,\dots,N_o-1\}$.
Then:
\begin{enumerate}
\item If $a\notin\{k_1,\dots,k_{N_e}\}$, then $\mathcal{C}_{a,b}|\phi\rangle=0$.
\item If $a\in\{k_1,\dots,k_{N_e}\}$ and $a+b\in\{k_1,\dots,k_{N_e}\}$, then $\mathcal{C}_{a,b}|\phi\rangle=0$.
\item If $a\in\{k_1,\dots,k_{N_e}\}$ and $a+b\notin\{k_1,\dots,k_{N_e}\}$, then $\mathcal{C}_{a,b}|\phi\rangle=\pm |\phi'\rangle$, where $|\phi'\rangle$ is the Slater determinant obtained from $|\phi\rangle$ by replacing the occupied orbital $a$ with $a+b$ and then re-ordering the creation operators in increasing order. The sign $\pm$ is determined by CAR re-orderings.
\end{enumerate}
\end{lemma}

\begin{proof}
Let $K=\{k_1,\dots,k_{N_e}\}$ be the set of occupied orbitals in $|\phi\rangle$.
\begin{enumerate}
\item If $a\notin K$, then $\hat c_a|\phi\rangle=0$ because $\hat c_a$ annihilates the state unless orbital $a$ is occupied. Hence $\mathcal{C}_{a,b}|\phi\rangle=\hat c^\dagger_{a+b}\hat c_a|\phi\rangle=0$.
\item If $a\in K$, then $\hat c_a|\phi\rangle$ removes the fermion in orbital $a$ and produces a nonzero $(N_e-1)$-particle Slater determinant (up to a sign determined by commuting $\hat c_a$ through the ordered creation operators). Acting with $\hat c^\dagger_{a+b}$ yields zero if $a+b\in K$ (Pauli exclusion), hence $\mathcal{C}_{a,b}|\phi\rangle=0$.
\item If $a\in K$ and $a+b\notin K$, then $\hat c_a|\phi\rangle\neq 0$ and $\hat c^\dagger_{a+b}$ creates a particle in an unoccupied orbital, giving a nonzero $N_e$-particle Slater determinant. Re-ordering the resulting creation operators into increasing order yields $|\phi'\rangle$ up to an overall sign from CAR. Thus $\mathcal{C}_{a,b}|\phi\rangle=\pm|\phi'\rangle$.
\end{enumerate}
\end{proof}

\begin{lemma}[Constructive map between two occupation configurations]
Let $S,T\subset\{0,1,\dots,N_o-1\}$ be two occupation sets with $|S|=|T|=N_e$, and let
$|S\rangle$ and $|T\rangle$ be the corresponding Slater determinants (with creation operators ordered increasingly).
Define the set differences
\begin{equation}
S\setminus T=\{ \text{orbitals occupied in $S$ but not in $T$}\},\qquad
T\setminus S=\{ \text{orbitals occupied in $T$ but not in $S$}\}.
\end{equation}
Then $|S\setminus T|=|T\setminus S|=:m\le N_e$, and there exist pairs $(s_j,t_j)$ with
$s_j\in S\setminus T$ and $t_j\in T\setminus S$ for $j=1,\dots,m$ such that
\begin{equation}
\left(\prod_{j=1}^m \hat c^\dagger_{t_j}\hat c_{s_j}\right)|S\rangle=\pm |T\rangle,
\end{equation}
and the intermediate vectors obtained after each partial product are all nonzero.
\end{lemma}

\begin{proof}
If $S=T$, take $m=0$ and the statement holds trivially. Assume $S\neq T$.
Pick any $t_1\in T\setminus S$. Since $|S|=|T|$, we must have $S\setminus T\neq \emptyset$;
choose any $s_1\in S\setminus T$. By construction, $s_1$ is occupied in $|S\rangle$ while
$t_1$ is unoccupied in $|S\rangle$. Hence, by Lemma~1,
\begin{equation}
\hat c^\dagger_{t_1}\hat c_{s_1}|S\rangle=\pm |S_1\rangle\neq 0,
\end{equation}
where
\begin{equation}
S_1=(S\setminus\{s_1\})\cup\{t_1\}.
\end{equation}
Now compare $S_1$ with $T$. The symmetric difference strictly decreases:
\begin{equation}
|S_1\triangle T|=|S\triangle T|-2,
\end{equation}
because we have removed one element from $S\setminus T$ and added one element from $T\setminus S$.
Iterating this procedure, after $m=|T\setminus S|$ steps we reach a set $S_m=T$.
At each step $j$, we choose $t_j\in T\setminus S_{j-1}$ (so $t_j$ is unoccupied in the current state)
and $s_j\in S_{j-1}\setminus T$ (so $s_j$ is occupied), which guarantees that each application
$\hat c^\dagger_{t_j}\hat c_{s_j}$ is nonzero by Lemma~1. Therefore,
\begin{equation}
\left(\prod_{j=1}^m \hat c^\dagger_{t_j}\hat c_{s_j}\right)|S\rangle=\pm |T\rangle.
\end{equation}
\end{proof}

\begin{theorem}[Generation of $\mathcal{H}_{N_e,N_o}$ by single-particle moves]
Fix any $|\phi_i\rangle\in\mathcal{H}_{N_e,N_o}$.
For integers $a,b$ with $a,a+b\in\{0,1,\dots,N_o-1\}$, define
\begin{equation}
\mathcal{C}_{a,b}=\hat c^\dagger_{a+b}\hat c_a.
\end{equation}
Then the linear span of the set
\begin{equation}
\left\{ \prod_{j=1}^{N_{\max}}\mathcal{C}_{a_j,b_j}|\phi_i\rangle \ \middle|\
N_{\max}\in\mathbb{Z},\ 0\le N_{\max}\le N_e\right\}
\end{equation}
equals $\mathcal{H}_{N_e,N_o}$.
\end{theorem}

\begin{proof}
Let $|\phi_f\rangle\in\mathcal{H}_{N_e,N_o}$ be any Fock basis state with occupation set $T$,
and let the occupation set of $|\phi_i\rangle$ be $S$. By Lemma~2, there exist pairs
$(s_j,t_j)$ with $s_j\in S\setminus T$ and $t_j\in T\setminus S$ for $j=1,\dots,m$,
where $m=|T\setminus S|\le N_e$, such that
\begin{equation}
\left(\prod_{j=1}^m \hat c^\dagger_{t_j}\hat c_{s_j}\right)|\phi_i\rangle=\pm|\phi_f\rangle.
\end{equation}
Each factor $\hat c^\dagger_{t_j}\hat c_{s_j}$ is of the form $\mathcal{C}_{a_j,b_j}$ by choosing
\begin{equation}
a_j=s_j,\qquad b_j=t_j-s_j,
\end{equation}
and the orbital-validity conditions in the theorem ensure $a_j,a_j+b_j\in\{0,\dots,N_o-1\}$.
Therefore every Fock basis vector $|\phi_f\rangle$ is contained (up to a nonzero sign) in the set
generated from $|\phi_i\rangle$ by products of $\mathcal{C}_{a,b}$ of length at most $N_e$.
Since the Fock basis spans $\mathcal{H}_{N_e,N_o}$, the linear span of the generated set equals
$\mathcal{H}_{N_e,N_o}$.
\end{proof}

\begin{lemma}[Triangular change of basis on the disk]
\label{lem:triangular_basis}
Consider the disk (symmetric gauge) with $N_o=2S+1$ orbitals labelled by $j=0,1,\dots,2S$.
For $m,n\ge 0$, define
\begin{equation}
\hat{\mathcal L}_{m,n}:=\sum_{i=1}^{N_e}(\hat b_i^\dagger)^m\hat b_i^n
=\sum_{j=0}^{2S} A_{m,n}(j) \hat c^\dagger_{j+m-n}\hat c_j,
\end{equation}
where
\begin{equation}
A_{m,n}(j)=
\begin{cases}
\displaystyle
\sqrt{\frac{j!}{(j-n)!}}\sqrt{\frac{(j-n+m)!}{(j-n)!}}, & j\ge n,\\
0,& j<n,
\end{cases}
\end{equation}
with the implicit constraint $0\le j+m-n\le 2S$.

Fix an integer shift $i$ with $|i|\le 2S$, and define operator families
\begin{equation}
l_n^{(i)}:=\hat{\mathcal L}_{n+i,n},\qquad
c_\sigma^{(i)}:=\hat c^\dagger_{\sigma+i}\hat c_\sigma,
\end{equation}
with index sets
\begin{equation}
\Sigma_i:=\{\sigma\mid 0\le \sigma\le 2S,\ 0\le \sigma+i\le 2S\},\qquad
\mathcal N_i:=\{0,1,\dots,|\Sigma_i|-1\}.
\end{equation}
Then
\begin{equation}
l_n^{(i)}=\sum_{\sigma\in\Sigma_i} M^{(i)}_{n,\sigma} c_\sigma^{(i)},
\qquad
M^{(i)}_{n,\sigma}:=A_{n+i,n}(\sigma),
\end{equation}
where, upon ordering $\Sigma_i$ increasingly,
\begin{equation}
M^{(i)}_{n,\sigma}=0\quad \text{for } \sigma<n,
\qquad
M^{(i)}_{n,n}=\sqrt{n!}\sqrt{(n+i)!}\neq 0.
\end{equation}
Hence $M^{(i)}$ is lower triangular with nonzero diagonal and therefore invertible. In particular,
\begin{equation}
\mathrm{span}\{l_n^{(i)}\}_{n\in\mathcal N_i}
=
\mathrm{span}\{c_\sigma^{(i)}\}_{\sigma\in\Sigma_i}.
\end{equation}
\end{lemma}

\begin{proof}
The expansion follows directly from the disk matrix elements of $(\hat b^\dagger)^{n+i}\hat b^n$.
The vanishing for $\sigma<n$ and the explicit nonzero diagonal entry at $\sigma=n$ imply that
$M^{(i)}$ is lower triangular with nonzero determinant, hence invertible.
\end{proof}

\begin{theorem}[Higher-spin deformations span the finite Hilbert space]
\label{thm:concise_span}
Fix $|\phi\rangle\in\mathcal H_{N_e,N_o}$. Then
\begin{equation}
\mathrm{span}\left\{
\prod_{t=1}^{N_{\max}}\hat{\mathcal L}_{m_t,n_t}|\phi\rangle
\ \middle|\ 
0\le N_{\max}\le N_e,\ m_t,n_t\ge 0
\right\}
=
\mathcal H_{N_e,N_o}.
\end{equation}
Consequently, for any ground state $|\psi_0\rangle$, every neutral excitation
$|\psi\rangle\in\mathcal G^\perp$ can be written as a (finite) superposition of geometric
deformations of $|\psi_0\rangle$ generated by $\hat{\mathcal L}_{m,n}$.
\end{theorem}

\begin{proof}
By the previously established result, finite products of
$\mathcal C_{a,b}=\hat c^\dagger_{a+b}\hat c_a$ generate $\mathcal H_{N_e,N_o}$ by linear span.
By Lemma~\ref{lem:triangular_basis}, for each fixed shift $b$ the operators
$\{\mathcal C_{a,b}\}$ and $\{\hat{\mathcal L}_{n+b,n}\}$ are related by an invertible linear map.
Replacing each $\mathcal C_{a,b}$ by its expansion in $\hat{\mathcal L}_{m,n}$ and expanding finite
products shows that the span generated by $\hat{\mathcal L}_{m,n}$ coincides with
$\mathcal H_{N_e,N_o}$. Restricting to $\mathcal G^\perp$ yields the statement for neutral excitations.
\end{proof}

For an infinite-dimensional Hilbert space, it requires a more formal construction using the Gelfand–Naimark–Segal (GNS) formalism, which we leave for future work to discuss about.


\end{document}